\documentclass[journal]{IEEEtran}

\usepackage{colortbl}
\usepackage{graphicx}
\usepackage{amsfonts}
\usepackage{amsmath}
\usepackage{paralist,amssymb}
\usepackage{arydshln}
\usepackage{mleftright}
\usepackage{braket}
\usepackage{multirow}
\usepackage{blkarray}
\usepackage{enumitem}
\usepackage{color}
\usepackage{amsthm}
\usepackage{nccmath}
\usepackage{cite}
\usepackage{fltpoint}
\usepackage{anyfontsize}
\usepackage{cancel}
\usepackage{caption}

\usepackage{algcompatible}
\usepackage{algorithm}

\makeatletter
\newcommand{\algmargin}{\the\ALG@thistlm}
\algnewcommand{\PARSTATE}[1]
{\STATE\parbox[t]{\dimexpr\linewidth-\algmargin}
{#1\vspace{0.1cm}}}

\usepackage{tikz}
\usetikzlibrary{decorations.pathreplacing}
\usepackage{physics}
\usepackage{booktabs}

\usepackage{hhline}

\newtheorem{theorem}{Theorem}
\newtheorem{corollary}{Corollary}
\newtheorem{lemma}{Lemma}

\newtheorem{definition}{Definition}

\newcommand{\F}{\mathbb{F}}
\newcommand{\Hs}{\mathcal{H}}
\newcommand{\ul}[1]{\underline{#1}}
\newcommand{\ol}[1]{\overline{#1}}

\DeclareMathOperator{\wt}{wt}
\DeclareMathOperator{\supp}{supp}
\DeclareMathOperator{\lcm}{lcm}

\usepackage[hidelinks,colorlinks=true,urlcolor=blue,linkcolor=blue,citecolor=red]{hyperref}

\begin{document}

\title{Concatenating Extended CSS Codes for\\Communication Efficient Quantum Secret Sharing}

\author{\vspace{0.3cm}
Kaushik Senthoor and Pradeep Kiran Sarvepalli
\thanks{The authors are with the Department of Electrical Engineering, Indian Institute of Technology Madras, Chennai 600036, India (email: ee15d042@ee.iitm.ac.in; pradeep@ee.iitm.ac.in).}
\thanks{This work was supported by the Department of Science and Technology, Govt. of India, under grant number DST/ICPS/QuST/Theme-3/2019/Q59.}
}

\maketitle

\begin{abstract}
Recently, a class of quantum secret sharing schemes called communication efficient quantum threshold secret sharing schemes (CE-QTS) was introduced.
These schemes reduced the communication cost during secret recovery.
In this paper, we introduce a general class of communication efficient quantum secret sharing schemes (CE-QSS) which include both threshold and non-threshold schemes.
We propose a framework for constructing CE-QSS schemes to generalize the earlier construction of CE-QTS schemes which was based on the staircase codes.
The main component in this framework is a class of quantum codes which we call the extended Calderbank-Shor-Steane codes.
These extended CSS codes could have other applications.
We derive a bound on communication cost for CE-QSS schemes.
Finally, we provide a construction of CE-QSS schemes meeting this bound using the proposed framework.
\end{abstract}


\section{Introduction}
A quantum secret sharing (QSS) scheme is a quantum cryptographic protocol for securely distributing a secret among multiple parties with quantum information.
In these schemes, the authorized sets of parties are allowed to recover the secret while unauthorized sets of parties are not allowed to have any information about the secret.
QSS schemes were first proposed by Hillery et al for sharing a classical secret\cite{hillery99}.
QSS schemes with multiple parties sharing a quantum secret were discussed by Cleve et al\cite{cleve99}.
The connections of QSS with quantum error correction codes were studied in \cite{cleve99,gottesman00,rietjens05,ps12,gheorghiu12,marin13,matsumoto14}.
Quantum secret sharing continues to be extensively studied\cite{karlsson99,smith99,bandyopadhyay00,nascimento01,tyc02,imai03,singh05,ogawa05,markham08,ps10,ben12,qin20} and experimentally implemented\cite{tittel01,wei13,pinnell20,ding22}.

In this paper, we focus on quantum secret sharing schemes which share a \textit{quantum} secret.
An important class of such QSS schemes is the $((t,n))$ quantum threshold schemes (QTS).
In these schemes any set of $t$ or more parties is authorized and any set of less than $t$ parties is unauthorized.
Cleve et al\cite{cleve99} first gave a general construction for QTS schemes.
Later Ogawa et al\cite{ogawa05} introduced ramp QSS schemes which are non-threshold QSS schemes and gave a construction.
In these schemes, any set of size $t$ is authorized but only sets of size $z$ or less have to be unauthorized.
Sets of sizes $z+1$ to $t-1$ have partial information about the secret.

The part of the encoded state given to a party is called the \textit{share} of that party.
For recovering the secret, the parties in an authorized set send their shares to a new party called the \textit{combiner}.
The combiner then recovers the secret with suitable operations.
The amount of quantum information sent to the combiner by the parties during secret recovery is called the \textit{communication cost}.
Alternatively, the parties in an authorized set can collaborate among themselves to recover the secret.
The communication cost for secret recovery will be slightly different in this case.

In the standard $((t,n))$ QTS schemes, when the combiner accesses $t$ or more parties, some $t$ parties have to send their complete shares.
It was shown by Gottesman\cite{gottesman00} that the size of each share in a QTS scheme has to be at least the size of the secret.
This implies that the communication cost is at least $t$ times the secret size.
Recently, \cite{senthoor19} proposed a class of QTS schemes, called communication efficient quantum threshold secret sharing (CE-QTS) schemes,  which reduced the communication cost for secret recovery.
In CE-QTS schemes, each party can also send just a part of its share instead of the full share.
In these schemes, the combiner can access any $t$ parties to recover the secret as in the QTS schemes.
Additionally, the combiner can also access any $d$ ($>t$) parties downloading only a part of the share from each party.
The communication cost while accessing $t$ parties is $t$ for each qudit of the secret. 
This reduces to $\frac{d}{d-t+1}$ while accessing $d$ parties.
For the maximum value of $d=2t-1$, this reduced communication cost is less than half the communication cost for $t$ parties.
The construction of CE-QTS schemes was motivated by the constructions for classical communication efficient secret sharing schemes in \cite{huang16,bitar18}.
The theory of CE-QTS was further developed in  \cite{senthoor20,senthoor22}.

As mentioned earlier, the size of each share in a QTS scheme is at least the size of the secret.
However, Ogawa et al\cite{ogawa05} proved that in the non-threshold ramp QSS schemes, the size of each share can be just $\frac{1}{t-z}$ times the secret size.
In general, non-threshold QSS schemes have less storage overhead compared to threshold schemes.
This motivates us to look for communication efficient non-threshold QSS schemes as well.
The contributions in this paper are listed below.

\textit{(i)}
In this paper, we construct communication efficient quantum secret sharing schemes which include both threshold and non-threshold schemes.
This work generalizes the construction in \cite{senthoor19} for communication efficient QTS schemes.
We refer to the proposed class of communication efficient quantum secret sharing schemes as CE-QSS schemes.

\textit{(ii)}
As the main result of this paper, we propose a general framework for constructing CE-QSS schemes.
This framework constructs CE-QSS schemes by concatenating a CSS code with another CSS code.
We analyze the parameters of the CE-QSS scheme thus obtained in terms of the classical linear codes used in the two CSS codes.
Our work is inspired from the framework for constructing classical communication efficient secret sharing schemes by Mart{\'\i}nez-Pe{\~n}as~\cite{martinez18}.

\textit{(iii)}
For the outer code used in the concatenation, we propose a class of CSS codes called \textit{extended CSS codes}.
The encoding in the extended CSS code is by extending an underlying CSS code with extra ancilla qudits.
We then characterize the extended CSS code as a QSS scheme.
This characterization builds upon Matsumoto\cite{matsumoto14} which showed how CSS codes can be characterized as QSS schemes.
The extended CSS codes described in this paper may be of independent interest.

\textit{(iv)}
The framework can provide many different constructions for CE-QSS schemes based on the family of classical linear codes being used in the CSS codes.
We use the family of generalized Reed-Solomon codes to obtain the construction for CE-QSS schemes similar to the construction for CE-QTS schemes in \cite{senthoor19}.
We also derive the bound on communication cost for secret recovery in CE-QSS schemes.
We see that the CE-QSS schemes from this construction are optimal in both the storage overhead and the communication cost.

%\subsection{Organization}
The paper is organized as follows.
In Section \ref{s:bg}, the necessary background on quantum secret sharing and CSS codes is given.
This section includes the definition of CE-QSS schemes.
In Section \ref{s:ce_qss_from_ecss}, we describe the extended CSS codes and then propose the framework to construct CE-QSS schemes from these codes.
In Section \ref{s:ce_qss_constrn_from_ecss}, we use this framework to provide a construction of optimal CE-QSS schemes based on generalized Reed-Solomon codes.

\section{Background}
\label{s:bg}

\subsection{Notation}
\label{ss:notation}
If $q$ is a prime power, then $\mathbb{F}_q$ denotes a finite field with $q$ elements. 
For any natural number $n$, we use the notation $[n]:=\{1,2,\ldots, n\}$.
For any $P\subseteq[n]$, its complement $[n]\setminus P$ is denoted as $\ol{P}$.
The collection of all subsets of $[n]$ given by $\{P:P\subseteq [n]\}$ is denoted as $2^{[n]}$.

If $M$ is an $\ell\times n$ matrix and $P\subseteq[n]$, then $M^{(P)}$ denotes the $\ell\times|P|$ submatrix of $M$ formed by taking the columns indexed by entries in $P$.
The $\ell\times\ell$ identity matrix is denoted as $I_\ell$.
For any linear code $C \subseteq \F_q^n$, we denote its generator matrix by $G_C$.
For linear codes $C_0$ and $C_1$ such that $C_1\subseteq C_0$, the term $G_{C_0/C_1}$ indicates the generator matrix of a complement of $C_1$ in~$C_0$.

We denote $\ket{x_1\,x_2\,\cdots\,x_\ell}$  by $\ket{\underline{x}}$ where $\underline{x}$ is the vector with entries from $(x_1,x_2,\hdots,x_\ell)$.
For a matrix $M\in \F_q^{a\times b}$, the notation $\ket{M}$ indicates the state $\ket{m_{11}m_{21}\hdots m_{a1}}$
$\ket{m_{12}m_{22}\hdots m_{a2}}$
$\hdots$
$\ket{m_{1b}m_{2b}\hdots m_{ab}}$ where $m_{ij}=[M]_{ij}$.

\begin{definition}[Minimum distance of a code]
For a linear code $C$, the minimum distance of the code is defined as
\begin{equation*}
\wt(C)=\min\{\wt(\underline{c})\ |\ \underline{c}\in C,\ \ul{c}\neq \ul{0}\}.
\end{equation*}
\end{definition}
\begin{definition}[Minimum distance of a nested code pair]
For linear codes $C_0$ and $C_1$ such that $C_1\subsetneq C_0$, the minimum distance of the nested code pair $(C_0,C_1)$ is defined as
\begin{equation*}
\wt(C_0\setminus C_1)=\min\{\wt(\underline{c})\ |\ \underline{c}\in C_0,\ \underline{c}\notin C_1\}.
\end{equation*}
\end{definition}

\subsection{Quantum secret sharing (QSS)}
\label{ss:bg_qss}
In this subsection, we review QSS  schemes and some of their properties.
A quantum secret sharing (QSS) scheme is a quantum cryptographic protocol where a secret is encoded and distributed among multiple parties such that only certain subsets of parties can recover the secret.
Subsets of parties which can recover the secret are called  \textit{authorized} sets (or \textit{qualified} sets) while subsets which have no information about the secret are called \textit{unauthorized} sets\ (or \textit{forbidden} sets).
We refer to the collection of all authorized sets as the \textit{access structure}, denoted by $\Gamma$.
The collection of all unauthorized sets is referred to as the \textit{adversary structure}, denoted by $\mathcal{A}$.
The access structure satisfies a monotonicity property in that if a set $A$ can recover the secret, then any set $A'\supseteq A$ can also recover the secret.
A QSS scheme is formally defined~as~follows.
\begin{definition}[QSS scheme]
\label{de:qss}
Consider two disjoint sets $\Gamma\subseteq 2^{[n]}$ and $\mathcal{A}\subseteq 2^{[n]}$.
An encoding of a quantum secret distributed among $n$ parties is defined as a QSS scheme for an access structure $\Gamma$ and an adversary structure $\mathcal{A}$ when the following conditions hold.
\begin{enumerate}
\item \textit{(Recovery)} For any $A\in\Gamma$, the secret can be recovered from parties in $A$.
\item \textit{(Secrecy)} For any $B\in\mathcal{A}$, the set of parties $B$ has no information about the secret.
\end{enumerate}
\end{definition}

A QSS scheme over $n$ parties is said to be {\em perfect} if every subset $P\subseteq[n]$ is either authorized or unauthorized \textit{i.e.} if $\Gamma\cup\mathcal{A}=2^{[n]}$.
In a {\em non-perfect} QSS scheme, there are sets of parties which are neither authorized nor unauthorized.
Such sets are called \textit{intermediate sets}.
Though an intermediate set of parties cannot recover the secret completely, it can retrieve partial information about the secret.
In this case, $\Gamma\cup\mathcal{A}\subsetneq 2^{[n]}$ and the collection of intermediate sets is given by $2^{[n]}\setminus(\Gamma\cup\mathcal{A})$.

In a classical secret sharing scheme, there may be two authorized sets which are disjoint.
This implies that there can be two copies of the secret which can be recovered independently.
However, the no-cloning theorem prohibits making copies of the quantum secret.
Hence any QSS scheme cannot have an access structure with two disjoint authorized sets.

\begin{lemma}
\label{lm:qss_forb}
In a QSS scheme, the complement of an authorized set is an unauthorized set  \textit{i.e.} if $A\in\Gamma$ then $\ol{A}\in\mathcal{A}$.
\end{lemma}

A QSS scheme is called a \textit{pure state QSS scheme} when any pure state secret is encoded into a pure state (shared among the $n$ parties).
The corresponding encoding is called a \textit{pure state encoding}.
A QSS scheme in which some pure state is encoded into a mixed state is called a \textit{mixed state QSS scheme}.

\begin{lemma}
\label{lm:pure_state_qss_forb}
In a pure state QSS scheme, the complement of a set is an unauthorized set if and only if the set is an authorized set \textit{i.e.} $A\in\Gamma\ \Leftrightarrow\ \ol{A}\in\mathcal{A}$.
\end{lemma}

Although the above two lemmas were shown by Cleve et al in the context of perfect schemes\cite[Corollary 8]{cleve99}, the proof therein follows for non-perfect schemes as well.

In any QSS scheme with $n$ parties and non-empty $\Gamma$, we can find some $0<t\leq n$ such that any set of $t$ or more parties is authorized.
Similarly, some $0\leq z<t$ can be found such that any subset of $z$ or less parties is unauthorized.

\begin{definition}[Parameters for QSS schemes]
\label{de:thresholds_for_qss}
For $0\leq z<t\leq n$, a QSS scheme is called a $((t,n;z))$ QSS scheme when the following conditions hold.
\begin{enumerate}
\item Any $P\subseteq[n]$ is authorized if $|P|\geq t$.
\item Any $P\subseteq[n]$ is unauthorized if $|P|\leq z$.
\end{enumerate}
\end{definition}
This definition of $((t,n;z))$ QSS schemes follow the definition of the ramp QSS schemes given in \cite{marin13,matsumoto14}.
This definition differs from the stricter definition of ramp QSS schemes in \cite{ogawa05}.
Whereas Definition~\ref{de:thresholds_for_qss} allows the possibility that some sets of size between $z+1$ to $t-1$ are authorized or unauthorized, these sets have to be intermediate (neither authorized nor unauthorized) for the $(t,t$$-$$z,n)$ ramp QSS scheme in \cite[Definition 1]{ogawa05}.
To denote this difference, we call the schemes in Definition \ref{de:thresholds_for_qss} simply as $((t,n;z))$ QSS schemes instead of referring to them as ramp QSS schemes.

If $z=t-1$, then there are no intermediate sets in a $((t,n;z))$ QSS scheme and the scheme becomes a $((t,n))$ quantum threshold secret sharing (QTS) scheme.

For a QSS scheme with non-empty $\Gamma$, let $t_{\min}$ be the minimum value of $t$ such that any set of $t$ or more parties is authorized.
Let $z_{\max}$ be the maximum value of $z$ such that any set of $z$ or less parties is unauthorized.
The following lemma gives the relation between $n$, $t_{\min}$ and $z_{\max}$.
\begin{lemma}[Bound on number of parties]
For any QSS scheme with non-empty $\Gamma$, the number of parties $n\leq t_{\min}+z_{\max}$.
If the QSS scheme is a pure state scheme, then equality holds.
\end{lemma}
\begin{proof}
By Lemma \ref{lm:qss_forb}, any set of size $n-t_{\min}$ is unauthorized which implies $z_{\max}\geq n-t_{\min}$.
Hence $n\leq t_{\min}+z_{\max}$.
Additionally, for a pure state QSS scheme, by Lemma \ref{lm:pure_state_qss_forb}, any set of $n-z_{\max}$ is authorized which implies $t_{\min}\leq n-z_{\max}$.
Hence $n\geq t_{\min}+z_{\max}$ as well.
\end{proof}
The above lemma says that for any pure state QSS scheme, $n=t_{\min}+z_{\max}$.
However, the converse need not be true.
There can be mixed state schemes with $n=t_{\min}+z_{\max}$.
In $((t,n))$
QTS schemes, we can see that $t_{\min}=t$ and $z_{\max}$ $=t-1$
which leads to the following lemma.
\begin{lemma}\cite[Theorem 2]{cleve99}
For any $((t,n))$ QTS scheme, the number of parties $n\leq 2t-1$.
If the QTS scheme is a pure state scheme, then equality holds.
\end{lemma}
A QSS scheme with the secret and all the shares having qudits of the same dimension $q$ is indicated by a subscript.
We refer to such a $((t,n;z))$ QSS scheme as $((t,n;z))_q$ QSS scheme.
The number of qudits (each of dimension $q$) in the secret is denoted by $m$.
The number of qudits in the $j$th share is denoted by $w_j$.
The lemma below gives the storage cost in distributing a secret in terms of number of the qudits.
\begin{definition}[Storage cost]
The storage cost for secret distribution in a $((t,n;z))_q$ QSS scheme equals $w_1+w_2+\cdots+w_n$.
\end{definition}

\subsection{Communication efficient QSS (CE-QSS)}
\label{ss:bg_ce_qss}
In this paper we are interested in designing QSS schemes which require a reduced communication overhead while reconstructing the secret. 
First we define the communication complexity of secret recovery for a given authorized set. 

\begin{definition}[Communication cost for an authorized set]
The communication cost for an authorized set $A\subseteq[n]$ in a $((t,n;z))_q$ QSS scheme is defined as
\begin{equation*}
\text{CC}_n(A)=\sum_{j\in A}h_{j,A}
\end{equation*}
where $h_{j,A}$ indicates the number of qudits sent to the combiner by the $j$th party during secret recovery from parties in $A$.
\end{definition}
Here we assume that for a given authorized set, which portion of an accessed share needs to be sent to the combiner is fixed a priori.
This assumption is necessary because there could be different ways to partition the shares from an authorized set for a successful secret recovery.

\begin{definition}[Communication cost for  $d$-sets]
The communication cost for a threshold $d$ such that $t\leq d\leq n$ in a $((t,n;z))_q$ QSS scheme is defined as
\label{de:comm-cost-d-set}
\begin{equation*}
\text{CC}_n(d)=\max_{A\subseteq[n],\,|A|=d}%{\substack{A\subseteq[n]\\|A|=d}}
\ \text{CC}_n(A).
\end{equation*}
\end{definition}

Based on the communication cost as described above, we define a communication efficient QSS scheme as follows.

\begin{definition}[Fixed $d$ CE-QSS]
\label{de:ce_qss}
Let $0\leq z<t<d\leq n$.
A $((t,n,d;z))_q$ communication efficient QSS scheme is a $((t,n;z))_q$ QSS scheme where
\begin{equation*}
\text{CC}_n(d)<\text{CC}_n(t).
\end{equation*}
\end{definition}
\noindent
Our previous work in \cite{senthoor19,senthoor20,senthoor22} has focused on communication efficient QSS schemes which are threshold schemes (with $z=t-1$), referred to as CE-QTS schemes.
However, the schemes in Definition \ref{de:ce_qss} include communication efficient QSS schemes which are non-threshold as well.
We refer to the schemes introduced here as CE-QSS schemes.

\subsection{CSS codes as QSS schemes}
\label{ss:bg_css_codes}
We briefly review the basics of CSS codes and the relevant results on QSS schemes based on CSS codes.
The CSS construction (for qubits) was independently proposed in \cite{calderbank96,steane96}.
For the generalization to qudits of prime power dimension, see \cite[Theorem 3]{grassl04} and \cite[Lemma 20]{ketkar06}.
\begin{definition}[CSS code]
Let $C_1\subsetneq C_0 \subseteq \F_q^n$ where $C_i$ is an $[n,k_i]_q$ linear code for $i\in\{0,1\}$. 
The CSS code of $C_0$ over $C_1$ is defined as the vector space spanned by the states
\begin{equation*}
\ket{\ul{x}+C_1}\ \equiv\ \frac{1}{\sqrt{|C_1|}}\ \sum_{\ul{y}\in C_1}\ket{\ul{x}+\ul{y}}
\end{equation*}
where $\ul{x}\in C_0$.
This code is denoted as CSS$(C_0,C_1)$.
It is an $[[n,k_0-k_1,\delta]]_q$ quantum code with distance
$\delta=\min\{\wt(C_0\setminus C_1),\wt(C_1^\perp\setminus C_0^\perp)\}$.
\end{definition}

Since $C_1\subsetneq C_0$, their generator matrices can be written as
\begin{equation}
G_{C_0}=
\left[
\begin{array}{c}
G_{C_0/C_1}\\
G_{C_1}
\end{array}
\right].
\end{equation}
Given $\ul{s}\in\F_q^{k_0-k_1}$, CSS$(C_0,C_1)$ encodes the state $\ket{\ul{s}}$ as 
\begin{equation}
\ket{\ul{s}}\ \mapsto\  \sum_{\ul{r}\in \F_q^{k_1}}
\ket{\,[\,G_{C_0/C_1}^T\ G_{C_1}^T\,]
\left[
\begin{array}{c}
\!\ul{s}\!\\
\!\ul{r}\!
\end{array}
\right]}
\ =\sum_{\ul{r}\in \F_q^{k_1}}
\ket{\,G_{C_0}^T
\left[
\begin{array}{c}
\!\ul{s}\!\\
\!\ul{r}\!
\end{array}
\right]}
\label{eq:css_encoding}
\end{equation}
where we dropped the normalization constant for convenience.

\begin{lemma}[QSS and QECC]\label{lm:qecc-qss-reln}
A pure state encoding of a quantum secret is a QSS scheme if and only if the encoded space corrects erasure errors on unauthorized sets and it corrects erasure errors on the complements of authorized sets.
\end{lemma}
Lemma~\ref{lm:qecc-qss-reln} is due to \cite[Theorem 1]{gottesman00} which discusses about the connection between quantum codes and QSS schemes.
Though \cite{gottesman00} discusses only about perfect QSS schemes, this lemma holds for non-perfect QSS schemes as well.
Using this lemma, we can obtain a QSS scheme from a CSS code. 
\begin{theorem}[QSS schemes from a CSS code]
\label{th:css_code_to_qss}
For any $0\leq z<t\leq n$ satisfying
\begin{subequations}
\begin{eqnarray}
t&\geq&n-\min\{\wt(C_0\setminus C_1),\wt(C_1^\perp\setminus C_0^\perp)\}+1
\\z&\leq&\min\{\wt(C_0\setminus C_1),\wt(C_1^\perp\setminus C_0^\perp)\}-1
\end{eqnarray}
\end{subequations}
the encoding in Eq.~\eqref{eq:css_encoding} gives a $((t,n;z))_q$ QSS scheme with
\begin{subequations}
\begin{gather}
m=k_0-k_1,
\\w_j=1\ \text{ for all }j\in[n],
\\\text{CC}_n(t)=t.
\end{gather}
\end{subequations}
\end{theorem}
\begin{proof}
Let each of the $n$ parties be given one encoded qudit.
An $[[n,k,\delta]]_q$ quantum code can correct erasures up to $\delta-1$ qudits.
Therefore, any set of $n-(\delta-1)$ or more parties is an authorized set.
By Lemma \ref{lm:qss_forb}, any set of size up to $\delta-1$ is unauthorized as well.
Thus the CSS code gives a $((t,n;z))_q$ QSS scheme.
Clearly, the share size $w_j=1$.
The communication cost is $t$ times the share size \textit{i.e.} CC$_n(t)=t$.
Since the code encodes $k_0-k_1$ qudits, the QSS scheme shares $m=k_0-k_1$ qudits.
\end{proof}

Theorem \ref{th:css_code_to_qss} discusses only about sets with $t$ or more parties and sets with $z$ or less parties.
It does not characterize whether a set of parties of size between $z+1$ to $t-1$ is authorized, unauthorized or intermediate.
Matsumoto \cite{matsumoto14} studies the QSS scheme from CSS code in more detail.
The following theorem from Matsumoto completely characterizes the access structure and the adversary structure.

\begin{theorem}[Authorized sets of QSS from CSS code] %\cite[Theorem~1]{matsumoto14}
\label{th:css_auth_set}
The $CSS(C_0,C_1)$ code gives a QSS scheme where $J\subseteq[n]$ is an authorized set if and only if both the conditions below hold.
\begingroup
\allowdisplaybreaks
\begin{subequations}
\begin{eqnarray}
\rank G_{C_0}^{(J)}-\rank G_{C_1}^{(J)}&=&\dim C_0-\dim C_1
\label{eq:cond_recover_s}
\\\rank G_{C_0}^{(\ol{J})}-\rank G_{C_1}^{(\ol{J})}&=&0
\label{eq:cond_disentangle_s}
\end{eqnarray}
\end{subequations}
\endgroup
\end{theorem}
The theorem above gives the necessary and sufficient set of conditions for a given set $J$ to be authorized.
Note that the encoding for the CSS code is a pure state encoding.
Therefore, by Lemma \ref{lm:pure_state_qss_forb}, the conditions \eqref{eq:cond_recover_s} and \eqref{eq:cond_disentangle_s} give the set of necessary and sufficient conditions for $\ol{J}$ to be an unauthorized set as well.
Now we can characterize $\Gamma$ and $\mathcal{A}$ as follows.
\begin{eqnarray}
\Gamma=&\{X\subseteq[n]:\text{Eq.~\eqref{eq:cond_recover_s} and \eqref{eq:cond_disentangle_s} hold true for }J=&X\}
\nonumber
\\\mathcal{A}=&\{Y\subseteq[n]:\text{Eq.~\eqref{eq:cond_recover_s} and \eqref{eq:cond_disentangle_s} hold true for }J=&\ol{Y}\}
\nonumber\\
\end{eqnarray}

\section{CE-QSS from extended CSS codes}
\label{s:ce_qss_from_ecss}
In this section, we give a classical communication efficient secret sharing (CE-SS) scheme based on the staircase structure and illustrate the idea behind the proposed framework for constructing CE-QSS schemes.
Then we describe the extended CSS codes and study its properties.
Finally we give the core result of this paper which is a framework to construct CE-QSS schemes by concatenating extended CSS code with CSS code.

Consider the following CE-SS scheme over $n=3$ parties with $t=2$ and $d=3$.
A secret symbol $s\in\F_5$ is encoded and each party is given two symbols from $\F_5$.

\vspace{0.2cm}
\begin{eqnarray*}
\begin{array}{c|c|c}
\hline
& \text{Layer 1}& \text{Layer 2} \\\hline
\text{Share 1}&s+r_1+r_2 & r_2+r_3\\
\text{Share 2}&s+2r_1+4r_2 & r_2+2r_3 \\
\text{Share 3}&s+3r_1+4r_2 & r_2+3r_3 \\\hline
\end{array}
\end{eqnarray*}
\vspace{-0.2cm}

The symbols $r_1$, $r_2$ and $r_3$ are chosen randomly from a uniform probability distribution over $\F_5$.
The set of first symbols with the three parties is called the first layer and the set of second symbols is called the second layer.

While accessing any two parties, each accessed party sends both its layers to the combiner thereby giving communication cost of 4 symbols.
However, while accessing all three parties, the combiner downloads symbols only from the first layer thereby giving a reduced communication cost of 3 symbols.
\begin{figure}[H]%[ht]
\begin{center}
\includegraphics[width=1.1\textwidth,trim=2cm 22.2cm 0 2cm]{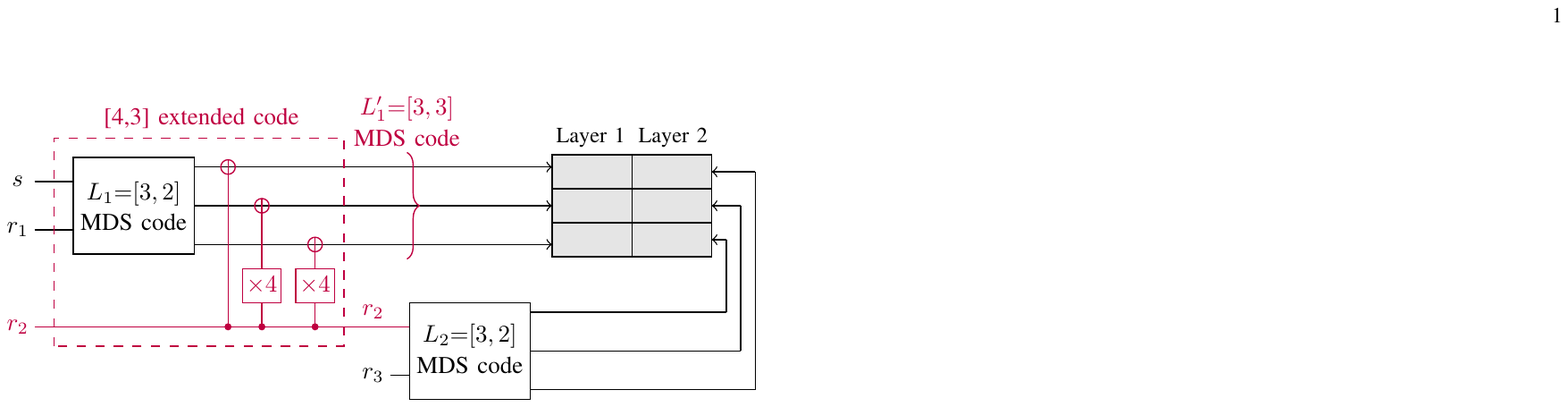}
\captionsetup{justification=justified}
\caption{Encoding circuit for the CE-SS scheme. The extension of the [3,2] code to [4,3] extended code is shown in purple.}
\label{fig:ce_ss_enc_ckt}
\end{center}
\end{figure}

The encoding in this scheme can be visualized as given in Fig. \ref{fig:ce_ss_enc_ckt}.
Whenever the combiner has access to any two parties, with the two symbols from the second layer, the code $L_2$ is first decoded to obtain $r_2$.
Now, the extension in the accessed two symbols in the first layer are inverted using $r_2$.
This gives the combiner two code symbols of the inner code $L_1$ which are then decoded to obtain the secret $s$.
The steps below explain the  secret recovery when the first two parties are accessed.
The cells in blue highlight the symbols whose values are modified during the recovery at that step.
%the symbols being decoded at each step.
\newcommand\columnsepminus{0.7cm}
\begin{center}
\begin{tabular}{c|c|c||c|c|}
\multicolumn{3}{c}{Layer 1}&\multicolumn{2}{c}{Layer 2}
\\\hhline{~----}
&\cellcolor{white}$s+r_1+r_2$
&\cellcolor{white}$s+2r_1+4r_2$
&\cellcolor{blue!10!white}$r_2+r_3$
&\cellcolor{blue!10!white}$r_2+2r_3$
\\\cline{2-5}
\multicolumn{5}{c}{}
\\[-\columnsepminus]
\multicolumn{5}{c}{}
\\\hhline{~----}
$\rightarrow$
&\cellcolor{blue!10!white}$s+r_1+r_2$
&\cellcolor{blue!10!white}$s+2r_1+4r_2$
&\cellcolor{blue!10!white}$r_2$
&\cellcolor{white}$r_3$
\\\cline{2-5}
\multicolumn{5}{c}{}
\\[-\columnsepminus]
\multicolumn{5}{c}{}
\\\hhline{~----}
$\rightarrow$
&\cellcolor{blue!10!white}$s+r_1$
&\cellcolor{blue!10!white}$s+2r_1$
& \cellcolor{white}$r_2$
& \cellcolor{white}$r_3$
\\\cline{2-5}
\multicolumn{5}{c}{}
\\[-\columnsepminus]
\multicolumn{5}{c}{}
\\\hhline{~----}
$\rightarrow$
&\cellcolor{white}$s$
&\cellcolor{white}$r_1$
&\cellcolor{white}$r_2$
&\cellcolor{white}$r_3$
\\\cline{2-5}
\multicolumn{5}{c}{}
\\
\end{tabular}
\end{center}
\vspace{-\baselineskip}\vspace{0.4cm}%
When the combiner accesses all three parties, the punctured code $L_1'$ from the extended code is decoded to obtain $s$.
\begin{center}
\begin{tabular}{c|c|c|c|}
\multicolumn{4}{c}{Layer 1}
\\\hhline{~---}
&\cellcolor{blue!10!white}$s+r_1+r_2$ & \cellcolor{blue!10!white}$s+2r_1+4r_2$ & \cellcolor{blue!10!white}$s+3r_1+4r_2$
\\\cline{2-4}
\multicolumn{4}{c}{}
\\[-\columnsepminus]
\multicolumn{4}{c}{}
\\\hhline{~---}
$\rightarrow$
&\cellcolor{white}$s$
&\cellcolor{white}$r_1$
&\cellcolor{white}$r_2$
\\\cline{2-4}
\multicolumn{4}{c}{}
\\
\end{tabular}
\end{center}
\vspace{-\baselineskip}

\subsection{Extended CSS codes}
\label{ss:ecss}
We will now describe the extended CSS code motivated from the classical extended code in the illustration above.
To construct a CE-QSS scheme encoding $m$ qudits, we take an $[[n,m]]_q$ CSS code and suitably extend it to an $[[n+e,m]]_q$ CSS code.
The encoding for this extended CSS code is done by suitably entangling some $e$ ancilla qudits with the $n$ encoded qudits from the $[[n,m]]_q$ CSS code.

Consider linear codes $F_0$ and $F_1$ of length $n$ over $\F_q$ with dimensions $f_0$ and $f_1$ respectively.
Also consider the matrix $G_E\in\F_q^{e\times n}$ whose row space is denoted by $E$.
Let $F_0$, $F_1$ and $G_E$ satisfy the following conditions.
\begin{enumerate}[label=N\arabic*.,ref=N\arabic*]
\item $F_1\subsetneq F_0$\label{cond:N1}
\item $F_0\cap E=\{\underline{0}\}$\label{cond:N2}
\end{enumerate}
\begin{figure}[ht]
\begin{center}
\includegraphics[width=1.3\textwidth,trim=-0.5cm 23.5cm 0cm 1.8cm]{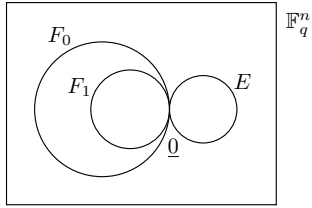}
\captionsetup{justification=justified}
\caption{Conditions on classical linear codes used for constructing the extended CSS code ECSS$(F_0,F_1,G_E)$}
\label{fig:ecss_codes}
\end{center}
\end{figure}

We can describe \ref{cond:N1}--\ref{cond:N2} in terms of generator matrices as
\begin{equation}
\left[
\begin{array}{c}
\vphantom{G_{F_0/F_1}}
\multirow{2}{*}{$G_{F_0}$}
\\\vphantom{G_{F_1}}
\\\hline G_{E}
\end{array}
\right]
=\left[
\begin{array}{c}
G_{F_0/F_1}
\\ G_{F_1}
\\\hline G_E
\end{array}
\right].
\end{equation}

With this choice of $F_0$, $F_1$ and $G_E$, we define the extended CSS code as follows.
\begin{definition}[Extended CSS code]
\label{de:ecss}
The extended CSS code ECSS$(F_0,F_1,G_E)$ is defined as the CSS$(C_0,C_1)$ code where
\begin{equation*}
G_{C_0}=\left[
\begin{array}{cc}
G_{F_0} & 0
\\G_E & I_e
\end{array}
\right],
\ G_{C_1}=\left[
\begin{array}{cc}
G_{F_1} & 0
\\G_E & I_e
\end{array}
\right].
\end{equation*}
\end{definition}
Note that the matrix $G_E$ here need not be of full rank for the above definition of extended CSS code.
From Theorem \ref{th:css_code_to_qss}, it is clear that this code gives a QSS scheme with $n+e$ qudits.

The encoding in ECSS$(F_0,F_1,G_E)$ can be written as
\begingroup
\allowdisplaybreaks
\begin{eqnarray}
\!\!\!\!
\ket{\ul{s}}
&\ \mapsto&
\sum_{\substack{
\ul{r}_1\in\,\F_q^{f_1}
\\\ul{r}_2\in\,\F_q^e}}
\ \ket{\,
\left[
\begin{array}{ccc}
G_{F_0/F_1}^T & G_{F_1}^T & G_E^T
\\0 & 0 & I_e
\end{array}
\right]
\left[
\begin{array}{c}
\!\ul{s}\!\\
\!\ul{r}_1\!\\
\!\ul{r}_2\!
\end{array}
\right]}\!\ \ \ \ 
\\&=&
\sum_{\substack{
\ul{r}_1\in\,\F_q^{f_1}
\\\ul{r}_2\in\,\F_q^e}}
\ \ket{\,[\,
G_{F_0/F_1}^T\ G_{F_1}^T\ G_E^T
\,]
\left[
\begin{array}{c}
\!\ul{s}\!\\
\!\ul{r}_1\!\\
\!\ul{r}_2\!
\end{array}
\right]}
\,\ket{\,\underline{r}_2\,}
\ \ 
\label{eq:ecss_encoding}
\end{eqnarray}
\endgroup
for any $\ul{s}\in\F_q^{f_0-f_1}$.
We refer to the first $n$ qudits in the encoded state as the \textit{original qudits} and the last $e$ qudits as the \textit{extension qudits}.
The block diagram for the encoding is given in Fig. \ref{fig:ecss_enc_ckt} where $\mathcal{U}(M)$ is the unitary operator such that, for any $\underline{x}\in\F_q^n$, $\underline{y}\in\F_q^{e}$ and $M\in\F_q^{n\times e}$,
\begin{equation}
\label{eq:ext_op}
\ket{\underline{x}}\ket{\underline{y}}\ \xrightarrow{\mathcal{U}(M)}\ \ket{\underline{x}+M\underline{y}}\ket{\underline{y}}.
\end{equation}

In the example CE-SS scheme discussed earlier, access to the extra parity ($r_2$) from the extended code depends on the number of parties accessed.
This extra parity encoded in the layer 2 is accessible to the combiner when any two parties are accessed.
However, it is inaccessible when only layer 1 of all the three parties are accessed.
Similarly, in the extended CSS code, we will consider the case when some of the extension qudits are known to be accessible or inaccessible to the combiner.

\begin{figure}[t]
\begin{center}
\includegraphics[width=1.15\textwidth,trim=1.5cm 23.5cm 0cm 1.8cm]{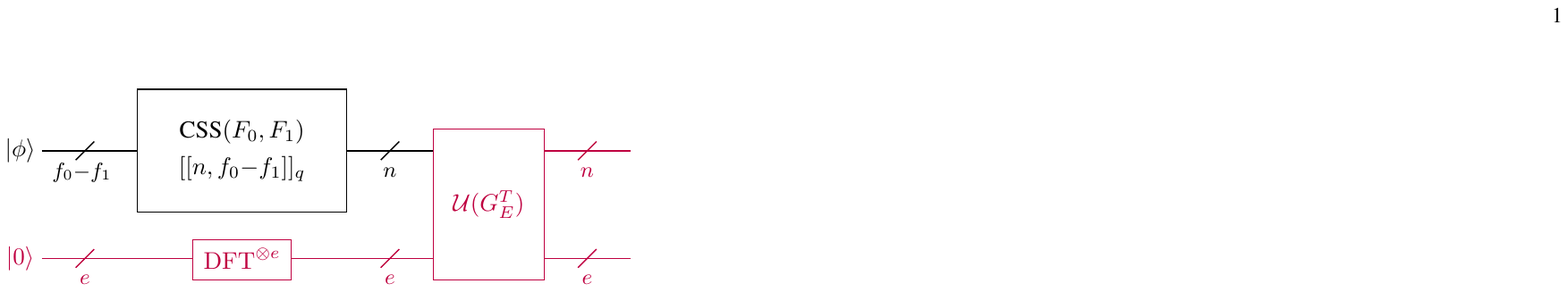}
\captionsetup{justification=justified}
\caption{Encoding circuit for ECSS$(F_0,F_1,G_E)$.
The extension of CSS$(F_0,F_1)$ to ECSS$(F_0,F_1,G_E)$ is shown in purple.}
\label{fig:ecss_enc_ckt}
\end{center}
\end{figure}

We will assume that the combiner has prior access to some $u$ of the extension qudits and does not have access to the remaining $v=e-u$ extension qudits.
In other words, in a QSS scheme from the extended CSS code, among the parties in $\{n+1,n+2,\hdots,n+e\}$, the combiner has prior access to some $u$ parties and no access to the remaining $v=e-u$ parties.
Under this assumption, we ask how many among the first $n$ parties is needed to recover the secret.

Without loss of generality, we take the already accessible parties as $\{n+1,n+2,\hdots n+u\}$ and the inaccessible parties as $\{n+u+1,n+u+2,\hdots n+e\}$.
Now, we can partition the rows in $G_E$ corresponding to these two sets as
\begin{equation}
G_E=\left[
\begin{array}{c}
G_U \\ G_V
\end{array}
\right]
\end{equation}
where $G_U$ is of size $u\times n$ and $G_V$ is of size $v\times n$.
The row spaces of $G_U$ and $G_V$ are indicated as $U$ and $V$ respectively.
The encoding in Eq.~\eqref{eq:ecss_encoding} can be rewritten as
\begin{align}
\label{eq:ecss_as_css_encoding}
&\ket{\ul{s}}
\ \mapsto\sum_{\substack{
\underline{r}_1\in\F_q^{f_1}
\\\underline{r}'_2\in\F_q^u
\\\underline{r}''_2\in\F_q^v
}}
\ket{\,
\left[
\begin{array}{cccc}
G_{F_0/F_1}^T & G_{F_1}^T & G_U^T & G_V^T
\\0 & 0 & I_u & 0
\\0 & 0 & 0 & I_v
\end{array}
\right]
\renewcommand{\arraystretch}{1.1}
\left[
\begin{array}{c}
\!\ul{s}\!\\
\!\ul{r}_1\!\\
\!\ul{r}'_2\!\\
\!\ul{r}''_2\!
\end{array}
\right]}.&
\nonumber
\\[-\baselineskip]&&
\end{align}

Now we analyze the access structure in the QSS scheme from extended CSS code when the extension qudits are either accessible or inaccessible to the combiner.
The set of authorized sets which include the accessed $u$ parties and exclude the inaccessible $v$ parties is given by
\begin{eqnarray}
\Omega_u=\left\{J\subseteq[n]\,{\bigg |}\ \parbox{5.2cm}{$J\cup\{n+1,n+2,\hdots n+u\}$ is an authorized set}\right\}\!.\ \ \ \ 
\end{eqnarray}
Specifically, we are interested in finding the threshold number $\tau_u$ of parties required out of the first $n$ parties to recover the secret with prior access to parties $n+1$, $n+2$, $\hdots$ $n+u$.
\begin{eqnarray}
\tau_u=\min\{\,\tau\,|\,\mbox{For all } J\subseteq[n]\text{ s.t. }|J|=\tau,\,\,J\in\Omega_u\}
\end{eqnarray}

We will use the conditions for a set to be authorized given in Theorem \ref{th:css_auth_set} to characterize $\tau_u$.
We also need the following two lemmas.
(Refer Appendix \ref{ap:proof_for_threshold_lemmas} for their proofs.)
These lemmas were used in Mart{\'\i}nez-Pe{\~n}as \cite{martinez18} to find the threshold number of parties needed for secret recovery in classical CE-SS schemes based on nested linear codes.

\begin{lemma}
\label{lm:bound_recover_s}
Let $1\leq\rho\leq n$.
For linear codes $C_0$ and $C_1$ such that $C_1\subsetneq C_0\subseteq\F_q^n$,
\begin{eqnarray}
\rank G_{C_0}^{(P)}-\rank G_{C_1}^{(P)}=\dim C_0-\dim C_1
\end{eqnarray}
for all $P\subseteq[n]$ such that $|P|=\rho$ if and only if
\begin{equation}
\rho\geq n-\wt(C_0\setminus C_1)+1.
\end{equation}
\end{lemma}

\begin{lemma}
\label{lm:bound_disentangle_s}
Let $1\leq\rho\leq n$.
For linear codes $C_0$ and $C_1$ such that $C_1\subsetneq C_0\subseteq\F_q^n$,
\begin{eqnarray}
\rank G_{C_0}^{(\ol{P})}-\rank G_{C_1}^{(\ol{P})}=0
\end{eqnarray}
for all $P\subseteq[n]$ such that $|P|=\rho$ if and only if
\begin{equation}
\rho\geq n-\wt(C_1^\perp\setminus C_0^\perp)+1.
\end{equation}
\end{lemma}

The following theorem studies the size of the sets in $\Omega_u$ using these two lemmas.
\begin{theorem}[Threshold with prior access to extension qudits]
\label{th:ecss_code_to_qss}
For every set $J\subseteq[n]$ such that $|J|=\tau$, the set $J\cup\{n+1,n+2,\hdots n+u\}$ is an authorized set in the QSS scheme from the ECSS$(F_0,F_1,G_E)$ code if and only if $\tau\geq\tau_u$ where
\begin{eqnarray}
&\hspace{-1cm}\tau_u= n-\min\{\wt((F_0+V)\setminus(F_1+V)),\hspace{1cm}&
\nonumber
\\\label{eq:tau_constr}
&\hspace{1.9cm}\wt((F_1+U)^\perp\setminus(F_0+U)^\perp)\}+1.&
\end{eqnarray}
\end{theorem}
\begin{proof}
Applying Theorem \ref{th:css_auth_set} to the encoding for the QSS scheme as given in Eq.~\eqref{eq:ecss_as_css_encoding}, the set $J\cup\{n+1,n+2,\hdots$ $n+u\}$ is authorized if and only if
\begingroup
\allowdisplaybreaks
\begin{subequations}
\begin{eqnarray}
\rank
\renewcommand{\arraystretch}{1.4}
\left[
\begin{array}{cc}
G_{F_0}^{(J)} & \mathbf{0}
\\G_U^{(J)} & I_u
\\G_V^{(J)} & \mathbf{0}
\end{array}
\right]
-\rank
\renewcommand{\arraystretch}{1.4}
\left[
\begin{array}{cc}
G_{F_1}^{(J)} & \mathbf{0}
\\G_U^{(J)} & I_u
\\G_V^{(J)} & \mathbf{0}
\end{array}
\right]
&=&f_0-f_1,
\nonumber
\\[-\baselineskip]
\\\rank
\renewcommand{\arraystretch}{1.5}
\left[
\begin{array}{cc}
G_{F_0}^{(\ol{J})} & \mathbf{0}
\\G_U^{(\ol{J})} & \mathbf{0}
\\G_V^{(\ol{J})} & I_v
\end{array}
\right]
-\rank
\renewcommand{\arraystretch}{1.5}
\left[
\begin{array}{cc}
G_{F_1}^{(\ol{J})} & \mathbf{0}
\\G_U^{(\ol{J})} & \mathbf{0}
\\G_V^{(\ol{J})} & I_v
\end{array}
\right]
&=&0
\end{eqnarray}
\end{subequations}
\endgroup
where $\ol{J}=[n]\setminus J$.
This set of conditions can be simplified to
\begingroup
\allowdisplaybreaks
\begin{subequations}
\begin{eqnarray}
\label{eq:cond_recover}
\rank
\renewcommand{\arraystretch}{1.4}
\left[
\begin{array}{c}
G_{F_0}^{(J)}
\\G_V^{(J)}
\end{array}
\right]
-\rank
\renewcommand{\arraystretch}{1.4}
\left[
\begin{array}{c}
G_{F_1}^{(J)}
\\G_V^{(J)}
\end{array}
\right]
&=&f_0-f_1,
\\\label{eq:cond_disentangle}
\rank
\renewcommand{\arraystretch}{1.5}
\left[
\begin{array}{c}
G_{F_0}^{(\ol{J})}
\\G_U^{(\ol{J})}
\end{array}
\right]
-\rank
\renewcommand{\arraystretch}{1.5}
\left[
\begin{array}{c}
G_{F_1}^{(\ol{J})}
\\G_U^{(\ol{J})}
\end{array}
\right]
&=&0.
\end{eqnarray}
\end{subequations}
\endgroup
The condition in Eq.~\eqref{eq:cond_recover} is same as Eq.~\eqref{eq:cond_recover_s} with $P=J$, $C_0=F_0+V$ and $C_1=F_1+V$.
By Lemma \ref{lm:bound_recover_s}, Eq.~\eqref{eq:cond_recover} holds true for all $J$ such that $|J|=\tau$ if and only if
\begin{equation}
\tau\geq n-\wt((F_0+V)\setminus(F_1+V))+1.
\end{equation}
Similarly, the condition in Eq.~\eqref{eq:cond_disentangle} is same as Eq.~\eqref{eq:cond_disentangle_s} with $P=J$, $C_0=F_0+U$ and $C_1=F_1+U$.
By Lemma \ref{lm:bound_disentangle_s}, Eq.~\eqref{eq:cond_disentangle} holds true for all $J$ such that $|J|=\tau$ if and only if
\begin{equation}
\tau\geq n-\wt((F_1+U)^\perp\setminus(F_0+U)^\perp)+1.
\end{equation}
Combining the two bounds above, we prove the theorem.
\end{proof}

So far we studied the extended CSS code under the assumption that some of the extension qudits are already accessible and others are inaccessible.
Now we evaluate the threshold $\tau_u$ in two specific cases, where the extension qudits are all accessible ($u=e$) or all inaccessible ($u=0$) to the combiner.

\begin{corollary}[Threshold with full access to extension qudits]
\label{co:tau_constr_ext_all}
When the combiner has prior access to all the extension qudits in ECSS$(F_0,F_1,G_E)$, the secret can be recovered from any $\tau$ qudits out of the $n$ original qudits if and only if $\tau\geq\tau_e$ where
\begin{flalign}
\tau_e=n-\min\{\wt(F_0\!\setminus\!F_1),\wt((F_1+E)^\perp\!\setminus\!(F_0+E)^\perp)\}+1.
\end{flalign}
\end{corollary}

\begin{corollary}[Threshold with no access to extension qudits]
\label{co:tau_constr_ext_none}
When the combiner has access to none of the extension qudits in ECSS$(F_0,F_1,G_E)$, the secret can be recovered from any $\tau$ qudits out of the $n$ original qudits if and only if $\tau\geq\tau_0$ where
\begin{flalign}
\tau_0=n-\min\{\wt((F_0+E)\!\setminus\!(F_1+E)),\wt(F_1^\perp\!\setminus\!F_0^\perp)\}+1.
\end{flalign}
\end{corollary}

\subsection{Concatenating extended CSS codes for CE-QSS}
\label{ss:ce_qss_from_ecss}
In this subsection, we give the main result of this paper.
We construct the CE-QSS scheme by concatenating an extended CSS code with another CSS code.
In the proposed CE-QSS scheme, the secret is first encoded using an extended CSS code.
The original qudits are stored in layer 1 and the extension qudits encoded using another CSS code are stored in layer 2.
We first give the conditions on the linear codes used to construct the extended CSS code and the CSS code used in layer 2.

Consider an $[n,b_0]$ linear code $B_0$ over $\F_q$.
Let $B_1$, $B_2$, $A_1$, $A_2$ and $E$ be linear codes of dimensions $b_1$, $b_2$, $a_1$, $a_2$ and $e$ respectively satisfying the following conditions.
\begin{enumerate}[label=M\arabic*.,ref=M\arabic*]
\item$B_2\subsetneq B_1\subsetneq B_0$\label{cond:M1}
\item$A_2\subseteq A_1\subsetneq B_0$ such that $B_0=B_1+A_1$ and $B_1\cap A_1=\{\ul{0}\}$ with $\dim A_2>0$\label{cond:M2}
\item$E\subseteq B_1$ such that $B_1=B_2+E$ and $B_2\cap E=\{\ul{0}\}$
\label{cond:M3}
\end{enumerate}  

Clearly $e=b_1-b_2$ and $a_1=b_0-b_1$.
The conditions \ref{cond:M1}--\ref{cond:M3} mentioned above can be visualized as shown in Fig.~\ref{fig:ceqss_enc_codes}.

\begin{figure}[H]
\begin{center}
\includegraphics[width=1.15\textwidth,trim=0.5cm 22.7cm 0cm 1.8cm]{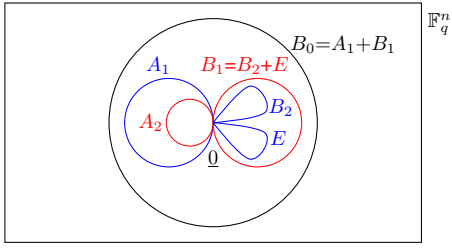}
\captionsetup{justification=justified}
\caption{Conditions on classical linear codes used for constructing the extended CSS code and the CSS code used in the CE-QSS scheme.
The classical codes used for the extended CSS code are marked in blue and those used for the CSS code are marked in red.}
\label{fig:ceqss_enc_codes}
\end{center}
\end{figure}

We can describe the conditions \ref{cond:M1}--\ref{cond:M3} in terms of generator matrices as
\begin{equation}
G_{B_0}=
\left[
\begin{array}{c}
\multirow{2}{*}{$G_{A_1}$}\\ \\ \hline \multirow{2}{*}{$G_{B_1}$} \\\ 
\end{array}
\right]
=
\left[
\begin{array}{c}
G_{A_1/A_2} \\ G_{A_2} \\ \hline G_{B_2} \\ G_{E}
\end{array}
\right].
\end{equation}

\textit{Encoding.}
The encoding for the proposed CE-QSS scheme is illustrated in Fig. \ref{fig:ceqss_enc}.
The $m=a_1 v_1$ qudits in the secret is partitioned into $v_1$ blocks of $a_1$ qudits each where each block is encoded by an ECSS$(A_1+B_2,B_2,G_E)$.
This encoding gives $v_1$ blocks each with $n+e$ qudits.
The $n$ original qudits from each of this block is stored in layer 1 of the $n$ parties.
The remaining $v_1 e$ extension qudits are rearranged into $v_2$ blocks of $a_2$ qudits each.
Then each of these blocks is encoded using a CSS$(A_2+B_1,B_1)$ code.
This encoding gives $v_2$ blocks each with $n$ encoded qudits which are stored layer 2 of the $n$ parties.

\begin{figure*}[ht!]
\begin{center}
\includegraphics[width=1.2\textwidth,trim=2cm 19.5cm 0cm 1.8cm]{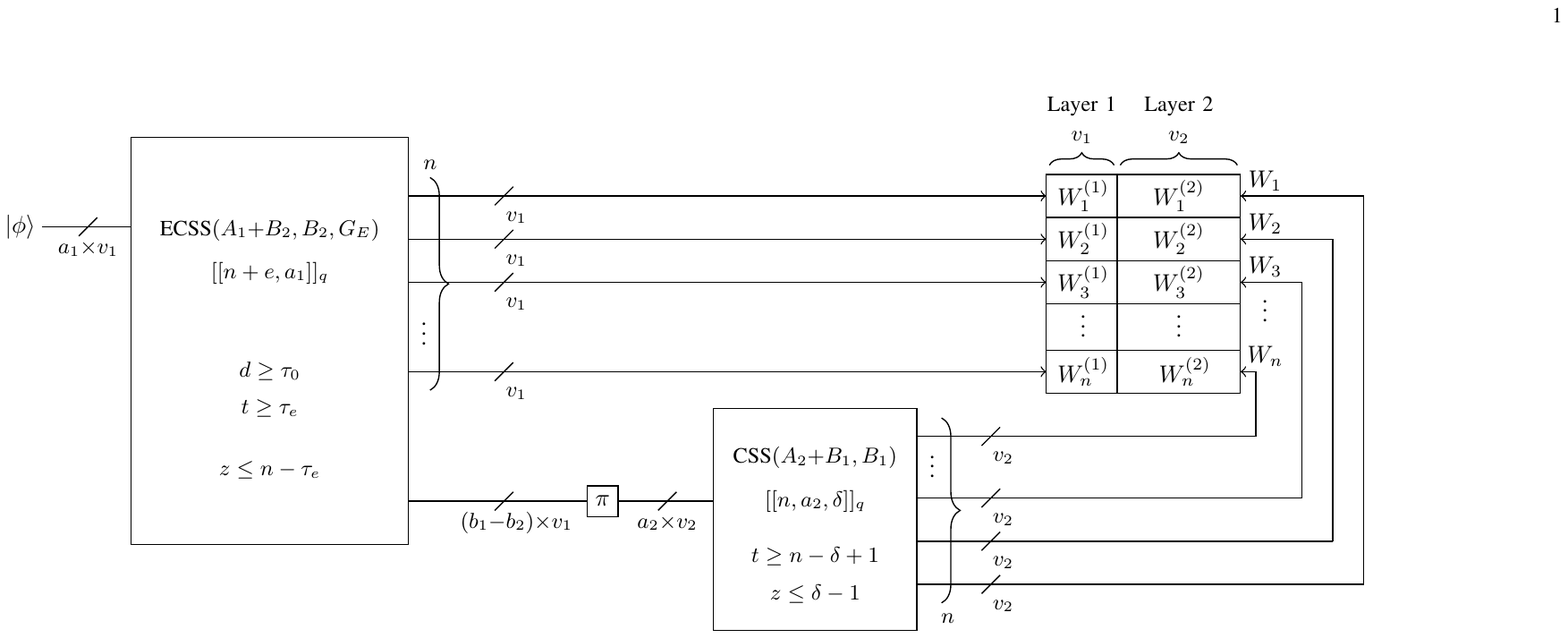}
\captionsetup{justification=justified}
\caption{Encoding for a $((t,n,d;z))_q$ CE-QSS scheme by concatenating extended CSS codes with CSS codes.}
\label{fig:ceqss_enc}
\end{center}
\end{figure*}

Extended CSS code provides the flexibility for the combiner to recover the secret from two different numbers of parties.
The secret can be recovered using the original qudits from layer 1 in one of the two following ways.
\begin{enumerate}[label=(\roman*)]
\item from any $d\geq\tau_0$ parties in layer 1
\item from any $t\geq\tau_e$ parties in layer 1 by also accessing the extension qudits stored in layer 2.
\end{enumerate}

The encoding of the extension qudits in layer 2 by a CSS code of distance more than $z$ is necessary to avoid eavesdropping by any $z$ parties.
If some $z$ parties were to get access to some information about the extension qudits from their layer 2, it is possible that this information can be used to recover some partial information about the secret from layer 1 of those $z$ parties.

We could have taken the original qudits of just a single extended CSS code for layer 1 and encoded its extension qudits with another CSS code for layer 2 to design the QSS scheme.
Instead we take $v_1$ instances of the extended CSS code and encode their extension qudits using $v_2$ instances of the CSS code.
This is because by varying $v_1$ and $v_2$, we get different (normalized) storage and communication costs in the CE-QSS scheme.
Then we can choose $v_1$ and $v_2$ to get the best possible storage and communication costs.

The encoding for the proposed CE-QSS scheme is defined using a message matrix with the staircase structure.
Using linear codes satisfying the conditions \ref{cond:M1}--\ref{cond:M3}, 
consider the encoding
\begin{eqnarray}
\label{eq:ce_qss_encoding}
\ket{S}\ \mapsto\sum_{\substack{
R_{1,1}\,\in\,\F_q^{b_2\times v_1}
\\R_{1,2}\,\in\,\F_q^{(b_1-b_2)\times v_1}
\\R_2\,\in\,\F_q^{b_1\times v_2}}}
\ \ket{\ G_{B_0}^T
\left[
\begin{array}{c:c}
\multirow{2}{*}{$S$} & \mathbf{0}
\\ & D_1
\\\hdashline R_{1,1} & \multirow{2}{*}{$R_2$}
\\R_{1,2}&
\end{array}
\right]\ }.
\end{eqnarray}
Here $S\in\F_q^{a_1\times v_1}$ indicates the basis state of the secret being encoded and $D_1\in\F_q^{a_2\times v_2}$ is the matrix formed by any rearrangement of the entries in $R_{1,2}$ where $a_2 v_2=(b_1-b_2)v_1$.
The message matrix $M$ used in the encoding in Eq.~\eqref{eq:ce_qss_encoding} and the sizes of its submatrices are given below.
\begin{eqnarray}
M_{\,b_0\times(v_1+v_2)}\ \ =
\begin{array}{ccc}
%row 1
&
\renewcommand{\arraystretch}{0.5}
\begin{array}{cc}
\ v_1 & v_2\\
\ \overbrace{} & \overbrace{}
\end{array}
&
\\%row 2 column 1
\begin{array}{cc}
\multirow{2}{*}{$a_1$} & \hspace{-0.3cm}\multirow{2}{*}{\Big\{}
\\&
\\b_2 & \hspace{-0.3cm}\{
\\e & \hspace{-0.3cm}\{
\end{array}
\!\!\!\!\!\!\!\!
&%row 2 column 2
\left[
\begin{array}{c:c}
\multirow{2}{*}{$S$} & 0
\\\cdashline{2-2} & D_1
\\\hdashline R_{1,1} & \multirow{2}{*}{$R_2$}
\\\cdashline{1-1} R_{1,2}&
\end{array}
\right]
&%row 2 column 3
\!\!\!\!\!\!\!\!
\begin{array}{cc}
&
\\\} & \hspace{-0.3cm}a_2
\\\multirow{2}{*}{\Big\}} & \hspace{-0.3cm}\multirow{2}{*}{$b_1$}
\\&
\end{array}
\end{array}
\end{eqnarray}

This encoding gives an $[[n(v_1+v_2),a_1 v_1]]$ quantum code.
For $1\leq j\leq n$, the $j$th party is given $v_1+v_2$ qudits corresponding to the $j$th row of the $n\times (v_1+v_2)$ matrix $G_{B_0}^T M$.
We refer to the first $v_1$ qudits in each party as layer 1 and the next $v_2$ qudits as layer 2.
We choose the smallest possible positive integers $v_1$ and $v_2$ such that $a_2 v_2=(b_1-b_2)v_1$ given by
\begin{eqnarray}
v_1=\frac{a_2}{\gcd\{a_2,b_1-b_2\}},
\ \ \ \ v_2=\frac{b_1-b_2}{\gcd\{a_2,b_1-b_2\}}.
\ \ 
\end{eqnarray}

\textit{Secret recovery.}
During secret recovery from $d$ parties, the combiner is given access to only layer 1 of the $d$ accessed parties.
The combiner then recovers the secret from the $d$ original qudits of the extended CSS code.
However, when the combiner has access to only $t<d$ parties, it downloads both the layers from these $t$ parties.
The combiner now first recovers the extension qudits by decoding the CSS code.
Then the combiner recovers the secret from $t$ original qudits from layer 1 and the extension qudits recovered from the second layer.
The following theorem gives the conditions when the encoding in Eq.~\eqref{eq:ce_qss_encoding} gives a CE-QSS scheme.
\begin{theorem}[CE-QSS using extended CSS codes]
\label{th:nested_code_ce_qss}
For any $0\leq z<t< d\leq n$ satisfying
\begin{subequations}
\label{eq:params_for_ce_qss}
\begin{eqnarray}
t&\geq&n-\text{min}\{\wt(A_2+B_1\setminus B_1),\nonumber
\\&&\hspace{1.45cm}\wt(A_1+B_2\setminus B_2),\wt(B_1^\perp\setminus B_0^\perp)\}+1\ \ 
\hspace{0.8cm}
\label{seq:t_for_ce_qss}
\allowdisplaybreaks\\d&\geq&n-\min\{\wt(B_0\!\setminus\!B_1),\wt(B_2^\perp\!\setminus\!(A_1\!+\!B_2)^\perp)\}+1\ \ 
\label{seq:d_for_ce_qss}
\allowdisplaybreaks\\z&\leq&\text{min}\{\wt(A_2+B_1\setminus B_1),\nonumber
\\&&\hspace{1cm}\wt(A_1+B_2\setminus B_2),\wt(B_1^\perp\setminus B_0^\perp)\}-1
\label{seq:z_for_ce_qss}
\allowdisplaybreaks\\\frac{d}{t}&<&\frac{a_2+b_1-b_2}{a_2}\,,
\label{seq:d_by_t_for_ce_qss}
\end{eqnarray}
\end{subequations}
the encoding in Eq.~\eqref{eq:ce_qss_encoding} gives a $((t,n,d;z))_q$ CE-QSS scheme with the following parameters.
\begingroup
\allowdisplaybreaks
\begin{subequations}
\label{eq:costs_for_ce_qss}
\begin{gather}
m=\frac{a_1 a_2}{\gcd\{a_2,b_1-b_2\}}
\\w_j=\frac{a_2+b_1-b_2}{\gcd\{a_2,b_1-b_2\}}\text{\ \ \ \ for all }j\in[n]
\\\text{CC}_n(t)=\frac{t(a_2+b_1-b_2)}{\gcd\{a_2,b_1-b_2\}}
\\\text{CC}_n(d)=\frac{d a_2}{\gcd\{a_2,b_1-b_2\}}
\end{gather}
\end{subequations}
\endgroup
\end{theorem}
\begin{proof}
\textit{(i) Recovery from $ d$ shares:}
When the combiner accesses any $d$ parties, each of these parties send the $v_1$ qudits from layer 1 to the combiner.

Let $D\subseteq[n]$ with $|D|=d$ give the set of accessed parties.
Layer 1 contains all the original qudits of the extended CSS code ECSS$(A_1+B_2,B_2,G_E)$ encoding the secret.
By Corollary \ref{co:tau_constr_ext_none}, this implies that the secret state $\ket{S}$ can be recovered from layer 1 of parties in $D$ if
\begin{equation}
d\geq n-\min\{\wt(B_0\!\setminus\!B_1),\wt(B_2^\perp\!\setminus\!(A_1\!+\!B_2)^\perp)\}+1.
\end{equation}
From Eq.~\eqref{seq:d_for_ce_qss}, it is clear that $d$ satisfies this condition and hence secret recovery is possible.
This implies that CC$_n(D)=d v_1$.
Since this is true for any $D$ such that $|D|=d$, from Definition \ref{de:comm-cost-d-set}, we obtain CC$_n(d)=d v_1.$
\vspace{0.2cm}
\\\textit{(ii) Recovery from $t$ shares:}
When the combiner accesses any $t$ parties, each of these $t$ parties sends all its $v_1+v_2$ qudits.
The encoded state in Eq.~\eqref{eq:ce_qss_encoding} can also be written as
\begin{eqnarray}
\hspace{-1.4cm}
\sum_{R_{1,1},\,R_{1,2},\,R_2}
\ \ket{G_{A_1}^T S+G_{B_2}^T\smash{R_{1,1}}+G_E^T\,\smash{R_{1,2}}}_{[n],1}
\nonumber
\\[-0.35cm]
\ket{G_{A_2}^T D_1+G_{B_1}^T R_2}_{[n],2}
\hspace{-0.8cm}
\end{eqnarray}
where the subscript to each ket indicates the set of parties and the layer containing the corresponding qudits.

Let $J\subseteq[n]$ with $|J|=t$ be the set of accessed parties.
Rearranging the qudits in the encoded state we obtain
\begin{flalign}
&\sum_{R_{1,1},\,R_{1,2},\,R_2}
\ket{\vphantom{G_{A_1}^{(J)}}\smash{{G_{A_1}^{(J)}}^{T} S+{G_{B_2}^{(J)}}^{T}R_{1,1}+{G_E^{(J)}}^{T}R_{1,2}}}_{J,1}&
\nonumber
\\[-0.3cm]
&\hphantom{\sum_{R_{1,1},\,R_{1,2},\,R_2}
\ket{\vphantom{G_{A_1}^{(J)}}\smash{{G_{A_1}^{(J)}}}}}
\ket{\vphantom{G_{A_2}^{(J)}}\smash{{G_{A_2}^{(J)}}^{T}D_1+{G_{B_1}^{(J)}}^{T}R_2}}_{J,2}&
\nonumber
\\[0.05cm]
&\hphantom{\sum_{R_{1,1},\,R_{1,2},\,R_2}\ \ \ }
\ket{\vphantom{{G_{A_1}^{(J)}}}\smash{{G_{A_1}^{(\ol{J})}}^{T} S+{G_{B_2}^{(\ol{J})}}^{T}R_{1,1}+{G_E^{(\ol{J})}}^{T}R_{1,2}}}_{\ol{J},1}&
\nonumber
\\[-0.05cm]
&\hphantom{\sum_{R_{1,1},\,R_{1,2},\,R_2}
\ket{\vphantom{G_{A_1}^{(J)}}\smash{{G_{A_1}^{(J)}}}}\ \ \ }
\ket{\vphantom{{G_{A_1}^{(J)}}}\smash{{G_{A_2}^{(\ol{J})}}^{T}D_1+{G_{B_1}^{(\ol{J})}}^{T}R_2}}_{\ol{J},2}.&
\end{flalign}
Layer 2 gives the encoded state from the CSS$(A_2+B_1,B_1)$ code encoding the extension qudits from the extended CSS code.
From Theorem \ref{th:css_code_to_qss}, these extension qudits can be recovered from layer 2 of parties in $J$ if
\begin{equation}
t\geq n-\text{min}\{\wt(A_2+B_1\setminus B_1),\wt(B_1^\perp\setminus (A_2+B_1)^\perp)\}+1.
\label{eq:cond_on_ext_recovery}
\end{equation}
Since $A_2+B_1\subseteq B_0$, we know that $\wt(B_1^\perp\setminus(A_2+B_1)^\perp)\geq\wt(B_1^\perp\setminus B_0^\perp)$.
Hence, from Eq.~\eqref{seq:t_for_ce_qss}, $t$ satisfies the condition in Eq.~\eqref{eq:cond_on_ext_recovery} and therefore the extension qudits can be recovered.

Recovering the extension qudits and discarding the remaining qudits from layer 2, we obtain
\begin{eqnarray}
\!\!\!\!\!\!\!\!\!
\sum_{R_{1,1},\,R_{1,2},\,R_2}
\ \ket{\vphantom{G_{A_1}^{(J)}}\smash{{G_{A_1}^{(J)}}^{T} S+{G_{B_2}^{(J)}}^{T}R_{1,1}+{G_E^{(J)}}^{T}R_{1,2}}}_{J,1}\ket{D_1}_{\,J,2}
\hspace{-0.7cm}
\nonumber
\\[-0.3cm]
\ket{\vphantom{G_{A_1}^{(J)}}\smash{{G_{A_1}^{(\ol{J})}}^{T}S+{G_{B_2}^{(\ol{J})}}^{T}R_{1,1}+{G_E^{(\ol{J})}}^{T}R_{1,2}}}_{\ol{J},1}
\nonumber\\
\end{eqnarray}
Since the matrix $D_1$ contains exactly the entries of $R_{1,2}$, the qudits can be rearranged to obtain
\begin{eqnarray}
\!\!\!\!\!\!\!
\sum_{R_{1,1},\,R_{1,2},\,R_2}
\ket{\vphantom{G_{A_1}^{(J)}}\smash{{G_{A_1}^{(J)}}^{T}\!S+{G_{B_2}^{(J)}}^{T}\!R_{1,1}+{G_E^{(J)}}^{T}\!R_{1,2}}}_{J,1}\ket{R_{1,2}}_{\,J,2}
\hspace{-0.7cm}
\nonumber
\\[-0.3cm]
\ket{\vphantom{G_{A_1}^{(J)}}\smash{{G_{A_1}^{(\ol{J})}}^{T}S+{G_{B_2}^{(\ol{J})}}^{T}R_{1,1}+{G_E^{(\ol{J})}}^{T}R_{1,2}}}_{\ol{J},1}
\nonumber\\
\end{eqnarray}
Layer 1 of the $n$ parties contains exactly the $n$ original qudits of the ECSS$(A_1+B_2,B_2,G_E)$ codes encoding the secret.
The combiner already has access to the extension qudits (indicated by $\ket{R_{1,2}}_{\,J,2}$) recovered from layer 2.
By Corollary \ref{co:tau_constr_ext_all}, the secret can be recovered from layer 1 of parties in $J$ if
\begin{equation}
t\geq n-\text{min}\{\wt(A_1+B_2\setminus B_2),\wt(B_1^\perp\setminus B_0^\perp)\}+1.
\end{equation}
Clearly $t$ from Eq.~\eqref{seq:t_for_ce_qss} satisfies this condition and hence secret recovery is possible from layers 1 and 2 of any $t$ parties.
This implies that the communication cost for secret recovery is CC$_n(J)=t(v_1+v_2)$.
Since this is true for any $J$ such that $|J|=t$, from Definition \ref{de:comm-cost-d-set}, we obtain CC$_n(t)=t(v_1+v_2)$.
\vspace{0.2cm}
\\\textit{(iii) Secrecy:}
We proved above that the secret can be recovered from any set of
\begin{equation}
n-\text{min}\{\wt(A_2+B_1\setminus B_1),
\wt(A_1+B_2\setminus B_2),
\wt(B_1^\perp\setminus B_0^\perp)\}+1
\end{equation}
or more parties.
By Lemma \ref{lm:qss_forb}, this implies that any set of $z$ parties is an unauthorized set for
\begin{eqnarray}
&&z\leq\text{min}\{\wt(A_2+B_1\setminus B_1),
\,\wt(A_1+B_2\setminus B_2),
\nonumber
\\&&\phantom{z\leq\text{min}\{\wt(A_2+B_1\setminus B_1),
\,\wt}\wt(B_1^\perp\setminus B_0^\perp)\}-1.
\ \ \ \ \ \ 
\end{eqnarray}
\textit{(iv) Communication efficiency:}
Since $a_2v_2=(b_1-b_2)v_1$,
\begin{eqnarray}
\text{CC}_n(t)
=t\text{$(v_1$$+$$v_2)$}
=t\left(\frac{\text{$a_2$$+$$b_1$$-$$b_2$}}{a_2}\right)v_1
>d v_1
=\text{CC}_n(d).
\nonumber\\
\end{eqnarray}
The inequality in the above expression is due to Eq.~\eqref{seq:d_by_t_for_ce_qss}.
\end{proof}

The $n$ original qudits from the extended CSS (or the $n$ encoded qudits from the CSS code) can be stored in the first layer (or the second layer) of the $n$ parties in any order.
The CE-QSS scheme thus obtained will have the same parameters as the original CE-QSS scheme.

\section{CE-QSS construction from the extended CSS framework}
\label{s:ce_qss_constrn_from_ecss}
So far we have provided the framework for constructing CE-QSS schemes using extended CSS codes.
For getting a specific construction of CE-QSS schemes based on this framework, we need to choose a family of classical linear codes to be used in the extended CSS codes.
In this section we use a family of classical MDS codes called generalized Reed-Solomon (GRS) codes\cite[Section 5.3]{huffman03} to provide an optimal construction for CE-QSS schemes.
First we look at the bounds on storage and communication costs for CE-QSS schemes.
The proof for the lemma below follows from \cite[Theorem 4]{ogawa05}. 
\begin{lemma}[Bound on storage cost]
\label{lm:qss_storage_cost}
For a $((t,n;z))_q$ QSS scheme encoding a secret of $m$ qudits,
\begin{equation}
\sum_{j=1}^{n}w_j\geq\frac{\displaystyle nm}{\displaystyle t-z}.
\end{equation}
\end{lemma}
\begin{proof}
The proof follows from \cite[Theorem 4]{ogawa05}.
\end{proof}

In the $(t,t$$-$$z,n)$ ramp QSS schemes defined in \cite[Defintion 1]{ogawa05}, any set of size between $z$$+$$1$ to $t$$-$$1$ should have neither full information nor zero information about the secret.
In contrast, the $((t,n;z))$ QSS schemes from Definition \ref{de:thresholds_for_qss} may have some sets of size between $z+1$ to $t-1$ as authorized or unauthorized.
However, \cite[Theorem 4]{ogawa05} and its proof for $(t,t$$-$$z,n)$ ramp QSS schemes hold true for $((t,n;z))_q$ QSS schemes as well.
For sake of completion, we give a detailed proof for $((t,n;z))_q$ QSS schemes in Appendix~\ref{ap:proof_for_storage_cost}.

For secret recovery from a given set of parties in a QSS scheme, we can design a truncated QSS scheme containing only these parties and storing only the parts of their shares sent to the combiner.
The communication cost for this set of parties in the original QSS scheme is then same as the storage cost of the truncated QSS scheme.
The following lemma uses this idea to derive a bound on the communication cost from an authorized set.
\begin{theorem}[Bound on communication cost]
\label{th:qss_comm_cost}
For a secret of size $m$ qudits, in a $((t,n;z))_q$ QSS scheme, the communication cost for secret recovery from an authorized set $A$ is bounded as
\begin{equation}
\text{CC}_n(A)\geq\frac{|A|m}{|A|-z}.\label{eq:cc-auth-set}
\end{equation}
\end{theorem}
\begin{proof}
Construct a new QSS scheme from the given $((t,n;z))_q$ QSS scheme by discarding the parties in $\ol{A}$ and the qudits not downloaded by the combiner from the parties in $A$.
Clearly, this is a $((|A|,|A|;z))_q$ QSS scheme encoding the same secret of $m$ qudits.
Let $\ell$ be the number of parties with no qudits in this truncated QSS scheme.
By dropping these $\ell$ parties, we obtain a $((|A|-\ell,|A|-\ell;z))_q$ QSS scheme.

The storage cost in the $((|A|-\ell,|A|-\ell;z))_q$ QSS scheme obtained is same as the communication cost CC$_n(A)$ in the given $((t,n;z))_q$ QSS scheme.
Then, by Lemma \ref{lm:qss_storage_cost},
\begin{gather}
\text{CC}_n(A)\geq\frac{(|A|-\ell)m}{|A|-\ell-z}\geq\frac{|A|m}{|A|-z}.
\end{gather}
\end{proof}

\begin{corollary}[Bound on communication cost for $d$-sets]
\label{co:qss_comm_cost_for_d}
For $t\leq d\leq n$, in a $((t,n;z))_q$ QSS scheme,
\begin{equation}
\text{CC}_n(d)\geq\frac{dm}{d-z}.
\end{equation}
\end{corollary}
\begin{proof}
The proof follows from Theorem \ref{th:qss_comm_cost}.
\end{proof}

Now we discuss a construction for CE-QSS schemes using the framework from Theorem \ref{th:nested_code_ce_qss}.
The construction is realised by choosing the linear codes in conditions from Eq.~\eqref{eq:params_for_ce_qss} to be generalized Reed-Solomon codes.
\begin{lemma}(Choosing GRS codes for CE-QSS)
\label{lm:choose_grs}
Let $B_0$, $B_1$, $B_2$, $A_1$, $A_2$ and $E$ be codes satisfying conditions \ref{cond:M1}--\ref{cond:M3} with generator matrices given by
\begin{eqnarray}
\left[
\begin{array}{c}
\!\!G_{A_1/A_2}\!\!\!
\\G_{A_2}\\G_{B_2}\\G_E
\end{array}
\right]
&=&
\left[
\begin{array}{cccc}
1 & 1 & \hdots & 1
\\x_1 & x_2 & \!\!\hdots\!\! & x_n
\\x_1^2 & x_2^2 & \!\!\hdots\!\! & x_n^2
\\\vdots & \vdots & \!\!\ddots\!\! & \vdots
\\x_1^{a_1-a_2-1} & x_2^{a_1-a_2-1} & \!\!\hdots\!\! & x_n^{a_1-a_2-1}
\\\hdashline
\vdots & \vdots & \!\!\ddots\!\! & \vdots
\\x_1^{a_1-1} & x_2^{a_1-1} & \!\!\hdots\!\! & x_n^{a_1-1}
\\\hdashline
\vdots & \vdots & \!\!\ddots\!\! & \vdots
\\x_1^{a_1+b_2-1} & x_2^{a_1+b_2-1} & \!\!\hdots\!\! & x_n^{a_1+b_2-1}
\\\hdashline
\vdots & \vdots & \!\!\ddots\!\! & \vdots
\\x_1^{b_0-1} & x_2^{b_0-1} & \hdots & x_n^{b_0-1}
\end{array}
\right]\!.
\nonumber
\\
\end{eqnarray}
Here $x_1, x_2,\hdots x_n$ are distinct non-zero constants from $\F_q$ where $q\geq n+1$ is a prime power.
Then the codes $B_0$, $B_1$, $B_2$, $A_1+B_2$, $A_2+B_1$ are generalized Reed-Solomon codes.
\end{lemma}

With the choice of linear codes as described above, we now discuss the CE-QSS construction based on Theorem \ref{th:nested_code_ce_qss}.
\begin{corollary}[Construction for CE-QSS from GRS codes]
\label{co:opt_ce_qss}
Let $q\geq n+1$ be a prime power.
For any $0<z<t<d\leq n=t+z$, choosing $B_0$, $B_1$, $B_2$, $A_1$, $A_2$ and $E$ as generalized Reed-Solomon codes as given in Lemma \ref{lm:choose_grs} for the encoding in Eq.~\eqref{eq:ce_qss_encoding} with
\begin{eqnarray}
&b_0=d\,,\ \ \ b_1=z\,,\ \ \ b_2=z-d+t\,,\ \ \ e=d-t\,,\ \ \ &
\nonumber
\\&a_1=d-z\,,\ \ \ a_2=t-z&
\end{eqnarray}
gives a $((t,n\,$$=\,$$t\,$$+\,$$z,d;z))_q$ QSS scheme with following parameters having optimal storage and communication costs.
\begingroup
\allowdisplaybreaks
\begin{subequations}
\begin{gather}
m=\lcm\{d-z,t-z\}
\\w_j=\frac{d-z}{\gcd\{d-z,t-z\}}\text{\ \ for all }j\in[n]
\\\text{CC}_n(t)=\frac{t(d-z)}{\gcd\{d-z,t-z\}}
\\\text{CC}_n(d)=\frac{d(t-z)}{\gcd\{d-z,t-z\}}
\end{gather}
\end{subequations}
\endgroup
\end{corollary}
\begin{proof}
By Lemma \ref{lm:choose_grs}, we know that the codes $B_0$, $B_1$, $B_2$, $A_1+B_2$ and $A_2+B_1$ are generalized Reed-Solomon codes.
Due to the Singleton bound, for any linear MDS code $L_0$ with another linear code $L_1\subsetneq L_0$, we know that $\wt(L_0\setminus L_1)\geq\wt(L_0)=n-\dim L_0+1$.
Now applying this bound, we see that conditions in Eq.~\eqref{eq:params_for_ce_qss} from Theorem \ref{th:nested_code_ce_qss} hold true for the chosen classical codes.
The parameters $m$, $w_j$, CC$_n(t)$ and CC$_n(d)$ follow from Eq.~\eqref{eq:costs_for_ce_qss}.

The CE-QSS scheme has optimal storage cost since the bound in Lemma \ref{lm:qss_storage_cost} is satisfied with equality.
The scheme has optimal communication costs as the values of CC$_n(t)$ and CC$_n(d)$ satisfy the bound in Corollary \ref{co:qss_comm_cost_for_d} with equality.
\end{proof}

The staircase codes based construction for $((t,n=2t-1,d))_q$ CE-QTS schemes given in \cite[Section III]{senthoor19} is a special case of the construction given in Corollary \ref{co:opt_ce_qss} where $z=t-1$.

\section{Conclusion}
In this paper, we introduced the class of communication efficient QSS schemes called CE-QSS which generalized CE-QTS schemes to include non-threshold schemes.
We proposed a framework based on extended CSS codes to construct CE-QSS schemes.
A specific construction using this framework has been provided to obtain CE-QSS schemes with optimal storage and communication costs.
As future work, we can look at constructions using this framework with families of linear codes other than GRS codes.
This work could be also extended to study CE-QSS schemes which generalize universal CE-QTS schemes from \cite{senthoor20}.

\bibliographystyle{IEEEtran}
%\bibliography{refs-ceqss-ecss-with-url}

\appendices
\section{Proofs for Lemmas \ref{lm:bound_recover_s} and \ref{lm:bound_disentangle_s}}
\label{ap:proof_for_threshold_lemmas}
We first define punctured codes and shortened codes and describe some of their properties.
For a general introduction to punctured codes and shortened codes, see \cite[Section 1.5]{huffman03}.

\begin{definition}[Punctured code]
For any $[n,k]$ code $C$ and $A\subseteq[n]$, the punctured code $C^A$ is defined as
\begin{equation*}
C^A=\{(c_i)_{i\in A}:(c_1,c_2,\hdots c_n)\in C\}.
\end{equation*}
\end{definition}

\begin{definition}[Shortened code]
For any $[n,k]$ code $C$ and $A\subseteq[n]$, the shortened code $C_A$ is defined as
\begin{equation*}
C_A=\{(c_i)_{i\in A}:c_j=0\ \forall j\in\ol{A},\ (c_1,c_2,\hdots c_n)\in C\}.
\end{equation*}
\end{definition}

\begin{lemma}[Properties of punctured codes \& shortened codes]
\ \\\vspace{-2\baselineskip}
\label{lm:prop_codes}
\begin{enumerate}[label=(\roman*)]
\item\label{lm:ap_lm_2}\cite[Lemma 1]{forney94}
$\dim C^A+\dim C_{\ol{A}}=\dim C$
\item\label{lm:ap_lm_3}\cite[Lemma 2]{forney94}
$\dim C^A+\dim (C^\perp)_A=|A|$
\item\label{lm:ap_lm_0}\cite[Section 1.5.1]{huffman03}
$\dim C^A=\rank\,G_C^{(A)}$
\end{enumerate}
\end{lemma}

\begin{lemma}
\label{lm:proof_inter_1}
Consider linear codes $C_0$ and $C_1$ such that $C_1\subsetneq C_0\subseteq\F_q^n$. For any $W\subseteq[n]$,
\begin{eqnarray}
(C_0)_W\supsetneq(C_1)_W
\ \ \Leftrightarrow\ \ 
\exists\,\ul{c}\in C_0\setminus C_1\text{ with }\supp(\ul{c})\subseteq W.
\nonumber\\
\end{eqnarray}
\end{lemma}
\begin{proof}
Assume $\exists\,\ul{c}=(c_1,c_2,\hdots c_n)\in C_0\setminus C_1$ with $\supp(\ul{c})\subseteq W$.
This implies $(c_i)_{i\in W}\in(C_0)_W\setminus(C_1)_W$ and hence $(C_0)_W\supsetneq(C_1)_W$.

To prove the converse, since $(C_0)_W\supsetneq(C_1)_W$, take some $\ul{c}'\in(C_0)_W\setminus(C_1)_W$.
Consider the vector $\ul{c}\in\F_q^n$ with entries in $W$ from those in $\ul{c}'$ and remaining entries zero.
Clearly, $\ul{c}\in C_0\setminus C_1$ and $\supp(\ul{c})\subseteq W$.
\end{proof}

\begin{lemma}
\label{lm:proof_inter_2}
Consider linear codes $C_0$ and $C_1$ such that $C_1\subsetneq C_0\subseteq\F_q^n$.
Let $1\leq w\leq n$.
Then
\begin{eqnarray}
&&(C_0)_W=(C_1)_W\ \ \forall\,W\in[n]\text{ such that }|W|=w
\ \ \ \ \ \nonumber
\\&&\phantom{(C_0)_W=()}\Leftrightarrow\ \wt(C_0\setminus C_1)>w.
\end{eqnarray}
\end{lemma}
\begin{proof}
By contraposition of Lemma \ref{lm:proof_inter_1}, for any $W\subseteq[n]$,
\begin{eqnarray}
(C_0)_W=(C_1)_W
\ \Leftrightarrow\ 
\nexists\,\ul{c}\in C_0\setminus C_1\text{ with }\supp(\ul{c})\subseteq W.
\nonumber\\
\end{eqnarray}
Considering only sets of size $w$, this implies
$(C_0)_W=(C_1)_W$ for all $W\subseteq[n]$ such that $|W|=w$ if and only if
\begin{eqnarray}
&&\nexists\,\ul{c}\in C_0\setminus C_1\text{ with }\supp(\ul{c})\subseteq W,
\nonumber
\\&&\phantom{\nexists\,\ul{c}\in C_0\setminus C_1\text{ w}}
\forall\,W\in[n]\text{ such that }|W|=w.\ \ 
\end{eqnarray}
This is same as the condition $\wt(C_0\setminus C_1)>w$.
\end{proof}

\begin{proof}[Proof for Lemma \ref{lm:bound_recover_s}]
By Lemma \ref{lm:proof_inter_2}, the condition $\rho\geq n-\wt(C_0\setminus C_1)+1$ holds if and only if for all $W\subseteq[n]$ such that $|W|=n-\rho$, $(C_0)_W=(C_1)_W$.
By taking $P=\ol{W}$, this is equivalent to the condition that for all $P$ such that $|P|=\rho$, $(C_0)_{\overline{P}}=(C_1)_{\overline{P}}$ \textit{i.e.}
\begingroup
\allowdisplaybreaks
\begin{eqnarray}
\dim (C_0)_{\overline{P}}&=&\dim (C_1)_{\overline{P}}
\\\label{eq:rec_cond_proof_eq_n1}
\dim C_0-\dim C_0^P&=&\dim C_1-\dim C_1^P
\\\dim C_0^P-\dim C_1^P&=&\dim C_0-\dim C_1
\\\label{eq:rec_cond_proof_eq_n2}
\rank G_{C_0}^{(P)}-\rank G_{C_0}^{(P)}&=&\dim C_0-\dim C_1.
\end{eqnarray}
\endgroup
The equivalences in Eq.~\eqref{eq:rec_cond_proof_eq_n1} and \eqref{eq:rec_cond_proof_eq_n2} follow from \ref{lm:ap_lm_2} and \ref{lm:ap_lm_0} in Lemma \ref{lm:prop_codes} respectively.
\end{proof}

\noindent
\begin{proof}[Proof for Lemma \ref{lm:bound_disentangle_s}]
By Lemma \ref{lm:proof_inter_2}, the condition $\rho\geq n-\wt(C_1^\perp\setminus C_0^\perp)+1$ holds if and only if for all $W\subseteq[n]$ such that $|W|=n-\rho$, 
$(C_1^\perp)_W=(C_0^\perp)_W$.
By taking $P=\ol{W}$, this is equivalent to the condition that for all $P$ such that $|P|=\rho$, $(C_1^\perp)_{\ol{P}}=(C_0^\perp)_{\ol{P}}$ \textit{i.e.}
\begingroup
\allowdisplaybreaks
\begin{eqnarray}
\dim (C_1^\perp)_{\ol{P}}&=&\dim (C_0^\perp)_{\ol{P}}
\\\dim C_1^{\ol{P}}&=&\dim C_0^{\ol{P}}
\label{eq:rec_cond_proof_eq_n3}
\\\rank G_{C_0}^{(\ol{P})}&=&\rank G_{C_1}^{(\ol{P})}.
\label{eq:rec_cond_proof_eq_n4}
\end{eqnarray}
\endgroup
The equivalences in Eq.~\eqref{eq:rec_cond_proof_eq_n3} and \eqref{eq:rec_cond_proof_eq_n4} follow from \ref{lm:ap_lm_3} and \ref{lm:ap_lm_0} in Lemma \ref{lm:prop_codes} respectively.
\end{proof}

\section{Bound on storage cost for QSS}
\label{ap:proof_for_storage_cost}
In this section, we derive the bound on storage cost given in Lemma~\ref{lm:qss_storage_cost}.
First, we derive the information-theoretic conditions for authorized and unauthorized sets.
Then we use these conditions to bound the entropy of a \textit{significant} set of parties.
A significant set is a set of parties adding which an unauthorized set becomes authorized. Finally we use this bound to prove the bound on storage cost.
This approach can be seen as the generalization of the results in \cite{imai03} and \cite[Section VII]{senthoor22} both of which focus on perfect QSS schemes.

We first recall some necessary concepts from quantum information theory.
For further reading on the topic, see \cite{nielsen00}.
Consider a quantum system $\mathcal{Q}$ defined over the Hilbert space $\mathcal{H}_\mathcal{Q}$ of dimension $N_\mathcal{Q}$.
Then the density matrix corresponding to $\mathcal{Q}$ can be defined as
\begin{equation}
\rho_\mathcal{Q}=\sum_{i=0}^{N_\mathcal{Q}-1}\lambda_i\ketbra{\phi_i}{\phi_i}
\end{equation}
where $\{\lambda_i\}$ gives the probability distribution in a measurement over some basis of orthonormal states $\{\ket{\phi_0},\ket{\phi_1},\hdots$ $\ket{\phi_{N_\mathcal{Q}-1}}\}$ over the Hilbert space $\Hs_\mathcal{Q}$ of dimension $N_\mathcal{Q}$.
The von Neumann entropy of $\mathcal{Q}$ is given by
\begin{equation}
\mathsf{S}(\mathcal{Q})=-\tr(\rho_\mathcal{Q}\ \text{log}\ \rho_\mathcal{Q})=-\sum_{i=1}^{N_\mathcal{Q}}\lambda_i\ \text{log}\,\lambda_i.
\end{equation}
The value $\mathsf{S}(\mathcal{Q})$ is bounded as $0\leq\mathsf{S}(\mathcal{Q})\leq\text{log}\,N_\mathcal{Q}$.
The value is maximum when $\mathcal{Q}$ is in the maximally mixed state ($\rho_\mathcal{Q}=I_{N_\mathcal{Q}}/N_\mathcal{Q}$) and zero when it is in pure state.

Consider the bipartite quantum system $AB$ whose density matrix $\rho_{AB}$ over the Hilbert space $\mathcal{H}_A\otimes\mathcal{H}_B$. 
Joint quantum entropy of $AB$ is defined as
\begin{equation}
\mathsf{S}(AB)=-\tr(\rho_{AB}\ \text{log}\ \rho_{AB}).
\end{equation}
It satisfies two important properties.
\begin{eqnarray}
\mathsf{S}(AB)\leq \mathsf{S}(A)+\mathsf{S}(B)
\label{eq:sub-addi}
\\\mathsf{S}(AB)\geq|\mathsf{S}(A)-\mathsf{S}(B)|
\label{eq:araki-lieb}
\end{eqnarray}
The property \eqref{eq:sub-addi} is called subadditivity and \eqref{eq:araki-lieb} is called the Araki-Lieb inequality.
The Araki-Lieb inequality implies that, if $AB$ is in pure state, then $\mathsf{S}(A)=\mathsf{S}(B)$.

Let $\mathcal{R}$ be the reference system over the Hilbert space of dimension $N_\mathcal{Q}$ such that the combined system $\mathcal{RQ}$ is in the pure state
\begin{equation}
\ket{\Phi_{\mathcal{RQ}}}=\sum_{i=0}^{N_\mathcal{Q}-1}\sqrt{\lambda_i}\ket{\phi_i}_\mathcal{R}\ket{\phi_i}_\mathcal{Q}.
\nonumber
\end{equation}
It is easy to see that $\mathsf{S}(\mathcal{Q})=\mathsf{S}(\mathcal{R})$ and $\mathsf{S}(\mathcal{RQ})=0$.
The following lemma discusses the condition for recovering the state of $\mathcal{Q}$ after a quantum operation acted upon it.

\begin{lemma}[Quantum data processing inequality \cite{schumacher96}]
Consider the system $\mathcal{Q}$ in an arbitrary quantum state $\rho_\mathcal{Q}$.
Let $\mathcal{R}$ be a reference system such that $\mathcal{RQ}$ is in pure state.
If $\mathcal{F}$ is a quantum operation which takes the state $\rho_\mathcal{Q}$ to another state $\rho_{\mathcal{Q}'}$ given by the system $\mathcal{Q}'$, then
\begin{eqnarray}
\mathsf{S}(\mathcal{Q})\geq\mathsf{S}(\mathcal{Q}')-\mathsf{S}(\mathcal{R}\mathcal{Q}')
\end{eqnarray}
with equality achieved if and only if the original state $\rho_\mathcal{Q}$ can be completely recovered from $\rho_{\mathcal{Q}'}$.
\label{lm:qdpi}
\end{lemma}

Let $\mathcal{S}$ be the quantum system corresponding to the secret in a QSS scheme with reference system $\mathcal{R}$.
Let $W_j$ indicate the $j$th share in the scheme.
The joint system of shares from a set of parties $A=\{j_1,j_2,\hdots j_{|A|}\}$ is given by
\begin{equation}
W_A=W_{j_1}W_{j_2}\hdots W_{j_{|A|}}.
\end{equation}
Based on Lemma \ref{lm:qdpi}, we can derive the condition for a set to be authorized (or unauthorized) in a QSS scheme.
\begin{lemma}[Condition for authorized set]
\label{lm:auth_set_qit}
In a QSS scheme with $n$ parties, $X\subseteq[n]$ is an authorized set if and only if
\begin{equation}
\mathsf{S}(\mathcal{R}W_X)=\mathsf{S}(W_X)-\mathsf{S}(\mathcal{R}).
\end{equation}
\end{lemma}
\begin{proof}
Let $\mathcal{E}$ be the encoding operation in the QSS scheme \textit{i.e.} $\rho_{W_{[n]}}=\mathcal{E}(\rho_\mathcal{S})$.
We can define a operation $\mathcal{E}_X$ such that 
\begin{equation}
\rho_{W_X}=\mathcal{E}_X(\rho_\mathcal{S})=\tr_{W_{[n]\setminus X}}\mathcal{E}(\rho_\mathcal{S}).
\end{equation}
Now, apply Lemma \ref{lm:qdpi} with $\mathcal{F}=\mathcal{E}_X$, $\mathcal{Q}=\mathcal{S}$ and $\mathcal{Q}'=W_X$ to obtain the result.
\end{proof}

\begin{lemma}[Condition for unauthorized set]
\label{lm:unauth_set_qit}
In a QSS scheme with $n$ parties, $Y\subseteq[n]$ is an unauthorized set if and only if
\begin{equation}
\mathsf{S}(\mathcal{R}W_Y)=\mathsf{S}(W_Y)+\mathsf{S}(\mathcal{R}).
\end{equation}
\end{lemma}
\begin{proof}
Consider a pure state QSS scheme constructed from the given QSS scheme by suitably adding an extra share $W_{n+1}$.
By Lemma \ref{lm:pure_state_qss_forb}, $Y$ is an unauthorized set if and only if $[n+1]\setminus Y$ is an authorized set.
By Lemma \ref{lm:auth_set_qit}, this implies that $Y$ is an unauthorized set if and only if
\begin{equation}
\mathsf{S}(\mathcal{S})=\mathsf{S}(W_{[n+1]\setminus Y})-\mathsf{S}(\mathcal{R}W_{[n+1]\setminus Y}).
\label{eq:temp_1}
\end{equation}
Since the $(n+1)$-party scheme is a pure state QSS scheme, $\mathcal{R}W_Y W_{[n+1]\setminus Y}$ is in pure state.
By Araki-Lieb inequality, this implies
$\mathsf{S}(W_{[n+1]\setminus Y})=\mathsf{S}(\mathcal{R}W_Y)
\text{ and }
\mathsf{S}(\mathcal{R}W_{[n+1]\setminus Y})=\mathsf{S}(W_Y)$.
Hence, the condition Eq.~\eqref{eq:temp_1} is equivalent to
\begin{eqnarray}
\mathsf{S}(\mathcal{S})&=&\mathsf{S}(\mathcal{R}W_Y)-\mathsf{S}(W_Y)
\\\mathsf{S}(\mathcal{R})&=&\mathsf{S}(\mathcal{R}W_Y)-\mathsf{S}(W_Y)
\\\mathsf{S}(\mathcal{R}W_Y)&=&\mathsf{S}(W_Y)+\mathsf{S}(\mathcal{R}).
\end{eqnarray}
\end{proof}

Now, we define a significant set of parties in a QSS scheme.
We first bound the joint entropy of shares from a significant set and then use this bound to prove the bound on storage cost.
\begin{definition}[Significant set]
For a QSS scheme, a set of parties $L$ is called a significant set when there exist a set of parties $Y$ such that
\begin{enumerate}[label=(\roman*)]
\item $Y$ is an unauthorized set and
\item $Y\cup L$ is an authorized set.
\end{enumerate}
\end{definition}

\begin{lemma}\cite[Theorem 3]{ogawa05}
\label{lm:sign_set}
For any significant set $L$ in a QSS scheme, 
\begin{equation}
\mathsf{S}(W_L)\geq\mathsf{S}(\mathcal{S}).
\end{equation}
\end{lemma}
\begin{proof}
Let $Y$ be an unauthorized set such that $Y\cap L=\{\}$ and $Y\cup L$ is an authorized set.
By Lemma \ref{lm:auth_set_qit},
\begin{equation}
\mathsf{S}(\mathcal{R}W_Y W_L)=\mathsf{S}(W_Y W_L)-\mathsf{S}(\mathcal{R}).
\end{equation}
Applying the Araki-Lieb inequality to $\mathcal{R}W_L W_Y$, we know that $\mathsf{S}(\mathcal{R}W_Y W_L)\geq\mathsf{S}(\mathcal{R}W_Y)-\mathsf{S}(W_L)$.
This implies
\begin{equation}
\mathsf{S}(\mathcal{R}W_Y)-\mathsf{S}(W_L)\leq\mathsf{S}(W_Y W_L)-\mathsf{S}(\mathcal{R}).
\end{equation}
Since $Y$ is an unauthorised set, applying Lemma~\ref{lm:unauth_set_qit}, we obtain
\begin{equation}
\mathsf{S}(\mathcal{R})+\mathsf{S}(W_Y)-\mathsf{S}(W_L)\leq\mathsf{S}(W_Y W_L)-\mathsf{S}(\mathcal{R}).
\end{equation}
Applying the subadditivity property on $W_Y W_L$, we obtain
\begin{eqnarray}
\mathsf{S}(\mathcal{R})+\mathsf{S}(W_Y)-\mathsf{S}(W_L)
&\leq&\mathsf{S}(W_Y)+\mathsf{S}(W_L)-\mathsf{S}(\mathcal{R})
\ \ \ \ \ 
\\2\,\mathsf{S}(\mathcal{R})
&\leq&2\,\mathsf{S}(W_L)
\\\mathsf{S}(\mathcal{R})
&\leq&\mathsf{S}(W_L)
\\\mathsf{S}(\mathcal{S})
&\leq&\mathsf{S}(W_L).
\end{eqnarray}
\end{proof}

\begin{proof}[Proof for Lemma \ref{lm:qss_storage_cost}]
Consider the $((t,n;z))_q$ QSS scheme.
For any $L\subseteq[n]$ such that $|L|=t-z$, we can always find a set $Y$ of $z$ parties such that $|L\cup Y|=t$.
This implies that any set of $t-z$ parties in a $((t,n;z))_q$ QSS scheme is a significant set.

While adding up the von Neumann entropies corresponding to all sets of size $t-z$, we obtain a lower bound
\begin{eqnarray}
\sum_{\substack{L\subseteq[n]\\\text{s.t. }|L|=t-z}}\mathsf{S}(W_L)
\ &\geq&\sum_{\substack{L\subseteq[n]\\\text{s.t. }|L|=t-z}}\mathsf{S}(\mathcal{S})
\label{eq:temp_2}
\\&=&\binom{n}{t-z}\,\mathsf{S}(\mathcal{S}).
\label{eq:temp_3}
\end{eqnarray}
The inequality in Eq.~\eqref{eq:temp_2} follows from Lemma \ref{lm:sign_set}.

Now, we find an upper bound for the term.
By subadditivity property, we know that,
\begin{eqnarray}
\label{eq:temp_3b}
\sum_{\substack{L\subseteq[n]\\\text{s.t. }|L|=t-z}}\mathsf{S}(W_L)
\ &\leq&\sum_{\substack{L\subseteq[n]\\\text{s.t. }|L|=t-z}}\sum_{j\in L}\,\mathsf{S}(W_j)
\end{eqnarray}
For a particular $j\in[n]$, there are $\binom{n-1}{t-z-1}$ combinations for set $L$ which contain $j$, out of the total $\binom{n}{t-z}$ possible combinations for set $L$.
Hence, by rearranging the entropy terms in RHS of Eq. \eqref{eq:temp_3b}, we obtain
\begin{eqnarray}
\sum_{\substack{L\subseteq[n]\\\text{s.t. }|L|=t-z}}\mathsf{S}(W_L)
\ \leq\binom{n-1}{t-z-1}\sum_{j=1}^{n}\,\mathsf{S}(W_j).
\label{eq:temp_4}
\end{eqnarray}
Comparing the bounds in Eq.~\eqref{eq:temp_3} and \eqref{eq:temp_4},
\begin{eqnarray}
\binom{n-1}{t-z-1}\sum_{j=1}^{n}\,\mathsf{S}(W_j)
\ &\geq&\ \binom{n}{t-z}\,\mathsf{S}(\mathcal{S})
\\\sum_{j=1}^{n}\,\mathsf{S}(W_j)
\ &\geq&\ \frac{n}{t-z}\ \mathsf{S}(\mathcal{S}).
\label{eq:temp_5}
\end{eqnarray}
The inequality in Eq.~\eqref{eq:temp_5} is similar to \cite[Theorem 4]{ogawa05}.
Since the $j$th share in a $((t,n;z))_q$ QSS scheme has $w_j$ qudits, the dimension of the share equals $\log q^{w_j}=w_j\log q$.
This implies
\begin{eqnarray}
\sum_{j=1}^{n}w_j\log q
&\geq&\frac{n}{t-z}\ \mathsf{S}(\mathcal{S})
\\\sum_{j=1}^{n}w_j
&\geq&\frac{n}{t-z}\,\frac{\mathsf{S}(\mathcal{S})}{\log q}.
\label{eq:temp_6}
\end{eqnarray}
Since $w_j$ is independent of the state of the secret $\rho_\mathcal{S}$, the above bound on the storage cost holds for all possible values of $\mathsf{S}(\mathcal{S})$.
We obtain the tightest bound possible from Eq.~\eqref{eq:temp_6} when $\mathsf{S}(\mathcal{S})$ is maximum.
This happens when the secret is in the maximally mixed state giving $\mathsf{S}(\mathcal{S})=\log q^m=m\log q$.
\begin{eqnarray}
\sum_{j=1}^{n}w_j&\geq&\frac{nm}{t-z}.
\end{eqnarray}
\end{proof}
More generally, if $N$ is the dimension of the secret in a $((t,n;z))$ QSS scheme, then the joint system of the $n$ shares is of dimension at least $N^{n/(t-z)}$.

\end{document}